\documentclass{article}
\usepackage{mathtools,amssymb,graphicx,times,comment}
\usepackage{amsmath,amsthm}
\usepackage{bbm}
  \usepackage{bm}
\usepackage{vaucanson-g}
\usepackage{float}
\usepackage{subfig}
\usepackage{tikz}
\usetikzlibrary{arrows}

\newcommand{\one}{\textbf{1}}

\newcommand{\R}{\mathbb{R}}

\newcommand{\al}{\alpha}

\newcommand{\xh}{\widehat{\x}}

\newcommand{\Xh}{\widehat{\X}}

\usepackage{dsfont}

\theoremstyle{theorem}
\newtheorem{propo}{Proposition}[section]
\newtheorem{thm}{Theorem}[section]
\newtheorem{lem}{Lemma}[section]

\newtheorem{ass}{Assumption}
\newtheorem{ex}{Example}
\theoremstyle{remark}
\newtheorem{remark}{Remark}
\theoremstyle{definition}

\newcommand{\Exp}{\mathbb{E}}               
\newcommand{\E}{\mathbb{E}}                    
\newcommand{\la}{{\lambda}}                     

\newcommand{\sig}{\sigma}

\newcommand{\nn}{\notag}

\newcommand{\X}{{X}}

\newcommand{\x}{\mathbf{x}}

\newcommand{\w}{\mathbf{w}}

\newcommand{\y}{\mathbf{y}}
\newcommand{\s}{\mathbf{s}}
\newcommand{\z}{\mathbf{z}}

\newcommand{\beq}{\begin{equation}}
\newcommand{\eeq}{\end{equation}}
\newcommand{\bea}{\begin{align}}
\newcommand{\eea}{\end{align}}

\usetikzlibrary{shadows} 
\usepackage{caption}
\usepackage{microtype}
\usepackage{booktabs} 
\usepackage{cite}
\usepackage[bookmarks=false]{hyperref}

\usepackage[accepted]{icml2020}
\usepackage{breqn} 

\icmltitlerunning{Quantized Decentralized Stochastic Learning over Directed Graphs}
\begin{document}
\twocolumn[
\icmltitle{Quantized Decentralized Stochastic Learning over Directed Graphs}
\begin{icmlauthorlist}
\icmlauthor{Hossein Taheri}{sb}
\icmlauthor{Aryan Mokhtari}{te}
\icmlauthor{Hamed Hassani}{penn}
\icmlauthor{Ramtin Pedarsani}{sb}
\end{icmlauthorlist}
\icmlaffiliation{sb}{Department of Electrical and Computer Engineering, University of California, Santa Barbara, Santa Barbara, CA, USA.}
\icmlaffiliation{te}{Department of Electrical and Computer Engineering, University of Texas at Austin, Austin, TX, USA.}
\icmlaffiliation{penn}{Department of Electrical and Systems Engineering, University of Pennsylvania, Philadelphia, PA, USA}
\icmlcorrespondingauthor{Hossein Taheri}{hossein@ucsb.edu}
\icmlkeywords{Decentralized Optimization, Quantization}
\vskip 0.3in
]
\printAffiliationsAndNotice{}

\begin{abstract}
We consider a decentralized stochastic learning problem where data points are distributed among computing nodes communicating over a directed graph. As the model size gets large, decentralized learning faces a major bottleneck that is the heavy communication load due to each node transmitting large messages (model updates) to its neighbors. To tackle this bottleneck, we propose the quantized decentralized stochastic learning algorithm over directed graphs that is based on the push-sum algorithm in decentralized consensus optimization. We prove that our algorithm achieves the same convergence rates of the decentralized stochastic learning algorithm with exact-communication for both convex and non-convex losses. Numerical evaluations corroborate our main theoretical results and illustrate significant speed-up compared to the exact-communication methods.
\end{abstract}
\section{Introduction}\label{sec:intro}

In many machine learning applications, optimization problems involving massive datasets often require the use of multiple computing agents to accelerate convergence. Moreover, in scenarios such as Federated Learning, each agent possesses its own local dataset. However, sharing data among agents is often impractical due to constraints such as privacy concerns and limited bandwidth. In such cases, computing nodes rely on their local data to perform (stochastic) gradient descent, exchanging parameters with other agents at each iteration to ensure convergence to the optimal solution of the global objective. More precisely, the goal of distributed learning is to  optimize a function $f:\R^d \rightarrow \R$  defined as the average of functions $f_i(\cdot)$, $i \in [n]$ of all computing nodes, i.e., 
\begin{align}\label{eq:gen_opt}
\xh := \arg \min_{\x\in\mathbb{R}^d} \left\{f(\x) =\frac{1}{n} \sum_{i=1}^n f_i(\x)\right\}.
\end{align}
Here, each local function $f_i(\cdot)$ can be considered as the loss incurred over the local data-set of node $i$: 
\begin{align}
f_i(\x) := \frac{1}{m_i}\sum_{j=1}^{m_i}\ell(\x,\zeta_j^i),
\end{align}
where  $\ell : \mathbb{R}^d\times\mathbb{R}^{d'}\rightarrow \mathbb{R}$ is the loss function, $m_i$ is the data-set size of node $i$, and $\zeta_j^i$ denotes the $d'$-dimensional local data points.  
Classical decentralized setting consists of two steps. First each computing node $i$ runs (stochastic) gradient descent algorithm on its local function $f_i(\cdot)$ using its local parameter $\x_i(t)$ and local data-set. Then the local parameters are exchanged between neighbor workers to compute a weighted average of neighbor nodes' parameters. Local parameters of every node $i$ at iteration $t+1$ is then obtained by a combination of the weighted average of neighboring nodes solutions and a negative descent direction based on local gradient direction. In particular, if we define $\x_i(t)$ as the decision variable of node $i$ at step $t$, then its local update can be written as 
\begin{align}\label{eq:decentralized}
    \x_i(t+1) = \sum_{j=1}^n a_{ij}\, \x_j(t) - \alpha(t) \nabla F_i\Big(\x_i(t), \zeta_{i,t}\Big).
\end{align}
Here $\alpha(t)\geq 0$ is the stepsize and $\nabla F_i$ is the stochastic gradient of function $f_i$ evaluated using a random sample of data-set of node $i$. Matrix $A$ represents the weights used in the averaging step and in particular $a_{ij}$ is the weight that node $i$ assigns to node $j$. It can be shown that under some conditions for the weight matrix $A$, the iterates of the update in~\eqref{eq:decentralized} converge to the optimal solution $\xh$ when local functions $f_i(\cdot)$ are convex \cite{yuan2016convergence} and to a first-order stationary point in the non-convex case \cite{lian2017can}. Most commonly in these settings, workers communicate over an \textit{undirected} connected graph $\mathcal{G} = \{V, E\}$, and to derive these theoretical results the weight matrix $A$ should have the following properties:

\begin{enumerate}
\item $A$ is doubly stochastic, i.e., $\one^TA = \one^T ;\; A \one = \one$, where $\one$ indicates an $n$-dimensional vector of all $1$'s;
\item $A$ is symmetric:  $A^T = A$;
\item The spectral gap of $A$ is strictly positive, i.e, the second largest eigenvalue of $A$ satisfies $\la_2(A)<1$.
\end{enumerate}
 It can be shown that these assumptions guarantee that the $t$-th power of $A$, i.e., $A^t$, converges to the matrix $\frac{1}{n}\one\one^T$ at a linear rate (i.e., exponentially fast); see, e.g., Lemma 5 in \cite{lian2017can}. This ensures the consensus among different workers in estimating the optimum solution $\xh$. However for \emph{directed} graphs, satisfying the first and second constraints are not generally possible. Over the last few years there have been several works to tackle decentralized optimization over directed graphs, e.g.,  \cite{blondel2005convergence} showed that for row-stochastic matrices with positive entries on the diagonal, the matrix $A^t$ converges to $\bm\one\bm\phi^T$ at a linear rate, where $\bm\phi$ is a stochastic vector. Based on this result \cite{nedic2014distributed} proposed the push-sum algorithm for decentralized optimization over (time-varying) directed graphs. The basic intuition is that the algorithm estimates the vector $\bm\phi$ by a vector $\y$ which is being updated among all workers in each iteration.

However, the mentioned algorithms for decentralized settings over directed graphs require exchanging the agents' model exactly and without any error in order to guarantee convergence to a desired solution. As the model size gets large, e.g., in deep learning algorithms, one can see that communication overhead of each iteration becomes the major bottleneck.  Parameter quantization is the major approach to tackle this issue. Although this approach might increase the overall  number of iterations, the reduction in the communication cost leads to an efficient algorithm for optimization or gossip problems with large model size. In this paper, we exploit the idea of compressing signals for communication to propose the first \textit{quantized algorithms} for gossip and decentralized learning algorithm over \textit{directed graphs}. The main contributions of this paper are summarized below:

\begin{itemize}

    \item We propose the first algorithm for communication-efficient gossip type problems and decentralized stochastic learning over directed graphs. Importantly, our algorithm guarantees converging to the optimal solution. 
    
    \item We prove that our proposed method converges at the same rate as push-sum with exact quantization. In particular for gossip type problems we show convergence in $O(\la^T)$. For stochastic learning problems with convex objectives over a directed network with $n$ nodes, we show that the objective loss converges to the optimal solution with the rate $O(\frac{1}{\sqrt{nT}})$. For non-convex objectives we show that squared norm of the gradient converges with the rate $O(\frac{1}{\sqrt{nT}})$, suggesting convergence to a stationary point with the optimal rate.
    
    \item The proposed algorithms demonstrate significant speed-up for communication time compared to the exact-communication method for gossip and decentralized learning in our experiments.
    
\end{itemize}

\noindent \textbf{Notation.} 
We use boldface notation for vectors and small and large letters for scalars and matrices respectively. $M^T$ denotes the transpose of the matrix $M$. Mean of rows of a matrix $M$, is denoted by $\overline{M}$. We use $[n]$ to denote the set of nodes $\{1,2,\cdots,n\}$. $\|\cdot\|$ denotes the $L_2$ norm of a matrix or vector depending on its argument. The $i$th row of the matrix $M$ is represented by $[M]_i$. The identity matrix of size $d\times d$ is denoted by $\mathbf{I}_d$ and the $d$ dimensional zero vector is denoted by $\mathbf{0}_d$. To simplify notation we represent the $n$-dimensional vector of all $1$'s as $\one$, where $n$ is the number of computing nodes.
\subsection{Prior Work}
\noindent \textbf{Gossip and decentralized optimization.} 
The consensus problems over graphs are generally called Gossip problems \cite{saber2003consensus,xiao2004fast,jadbabaie2003coordination}. In the gossip type problems, the goal of each node is to reach the average of initial values of all nodes over an undirected graph. Over the last few years there have been numerous research works which consider the convergence of decentralized optimization for undirected graphs \cite{nedic2009distributed,yuan2016convergence,shi2015extra,pu2020push}. In particular,
\cite{nedic2009distributed,yuan2016convergence} prove the convergence of decentralized algorithm for convex and Lipschitz functions, while \cite{lian2017can} prove the convergence of stochastic gradient descent for non-convex and smooth loss functions with the rate $O(\frac{1}{\sqrt{nT}})$. \cite{shi2015extra} propose the EXTRA algorithm which using gradient descent converges at a linear rate for strongly-convex losses.

\noindent \textbf{Decentralized learning over directed graphs.}
The first study of push-sum scheme for gossip type problems in directed graphs is discussed in \cite{kempe2003gossip}. Authors in \cite{tsianos2012push} extend this method to decentralized optimization problems and show the convergence of push-sum for convex loss functions. \cite{nedic2014distributed} extend these results to time-varying directed uniformly strongly connected graphs. Convergence of push-sum protocol for non-convex settings and for asynchronous communication are discussed in \cite{assran2018stochastic,assran2020asynchronous}. General algorithms for reaching linear convergence rate in decentralized optimization e.g., EXTRA can be combined with the push-sum algorithm to obtain similar results for strongly-convex objective function in directed graphs. See the following works for interesting discussions regarding this topic  \cite{zeng2015extrapush,xi2017distributed,xi2018linear,xin2018linear, xi2017dextra, xi2015linear, nedic2017achieving}. 

\noindent \textbf{Quantized decentralized learning.} 
Over the last few years there has been a surge of interest in studying communication efficient algorithms for distributed/decentralized optimization \cite{nedic2009distributed_quant, reisizadeh2019robust, alistarh2017qsgd,doan2018accelerating, koloskova2019decentralized,stich2018local,karimireddy2019error,wang2019slowmo,bernstein2018signsgd}. In \cite{nedic2009distributed_quant} authors discuss the effect of quantization in decentralized gossip problems using quantization with constant variance. However they show that these algorithms converge to the average of the initial values of the agents within some error. Authors in \cite{koloskova2019decentralized} propose a quantized algorithm for decentralized optimization over undirected graphs. Here we employ the push-sum protocol to extend this quantized scheme for directed graphs. On a technical level, we build upon and extend the theoretical results of \cite{lian2017can, nedic2014distributed, koloskova2019decentralized}.

\section{Network Model}
We consider a \emph{directed} graph $\mathcal{G} = \{V,E\}$, where $V$ is the set of nodes and $E$ denotes the set of directed edges of the graph. We say there is a link from node $i$ to node $j$ when $(i, j) \in E$. Indeed, as the graph is directed this does not guarantee that there is also a link from $j$ to $i$, i.e., $(j,i) \in E$. The sets of in-neighbors and out-neighbors of node $i$ are defined as:
\begin{align*}
\mathcal{N}_{i}^{\mathrm{in}} := \{j:(j, i) \in E\} \cup\{i\}, \\ \mathcal{N}_{i}^{\rm{out }} := \{j:(i, j) \in E\} \cup\{i\}.
\end{align*}
We denote  $d^{\rm \,out}_{i} := \left|\,\mathcal{N}_{i}^{\rm {\,out }}\right|$ to be the out-degree of node $i$. Throughout this paper we assume that $\mathcal{G}$ is strongly connected.
\begin{ass}[Graph structure]\label{assumption:graph} 
Graph $\mathcal{G}$ of workers is strongly connected. 
\end{ass}
Additionally we assume that the weight matrix has non-negative entries and each node uses its own parameter as well as its in-neighbors. This implies that all vertices of graph $\mathcal{G}$ have self-loops. Also we assume that the weight matrix is column stochastic. 
\begin{ass}[Weight matrix]\label{assumption:matrix}
Matrix $A$ is column stochastic, all entries are non-negative and all entries on the diagonal are positive. 
\end{ass}

Given the above assumptions, we next state a key result from \cite{nedic2014distributed, zeng2015extrapush} that will be useful in our analysis. 

\begin{propo}\label{lem:propo}
Let Assumptions \ref{assumption:graph} and \ref{assumption:matrix} hold and let $A$ be the corresponding weight matrix of workers in a graph $\mathcal{G}$. Then, there exist a stochastic vector $\bm\phi\in\mathbb{R}^n$, and constants $0<\la<1$ and $C>0$ such that for all $t\ge0$: 
\begin{align}\label{eq:matrix}
\quad \Big\|A^t-\bm\phi \mathbf{1}^T\Big\|\le C\la^t.
\end{align}
Moreover there exists a constant $\delta>0$ such that for all $i \in [n]$ and $t\ge1$ 
\begin{align}\label{eq:delta}
\left[A^t\mathbf{1}\right]_i\ge \delta.
\end{align}
\end{propo}
Note that the column-stochastic property of the weight matrix is considerably weaker than double-stochastic property. As explained in the next example, each computing node $i$ can use its own out-degree to form the $i$'th column of weight-matrix. Thus the weight matrix can be constructed in the decentralized setting without each node knowing $n$ or the structure of the graph. 
\begin{ex}\label{ex:graph}
Consider a strongly connected network of $n$ computing nodes where $a_{ij} = \frac{1}{d_{j}^{\,\rm{out}}}$ for all $i,j \in [n]$, and each node has a self-loop. It is straight-forward to see that $A$ is column stochastic and all entries on the diagonal are positive. Therefore the constructed weight matrix satisfies Assumption \ref{assumption:matrix}. 
\end{ex}

\section{Push-sum for Directed Graphs}
Before explaining our main contributions on quantized decentralized learning, we discuss gossip or consensus algorithms over directed graphs. Consensus algorithms in the decentralized setting are denoted as gossip algorithms. In this problem, workers are exchanging their parameters $\x_i(t)$ at time $t$ over a connected graph. The goal is to reach the average of initial values of all nodes, i.e., $\bar{X}(1)$, at every node, guaranteeing consensus among workers. The gossip algorithm is based on the weighted average  of parameters of neighboring nodes, i.e.,$
 \x_i(t+1) = \sum_{j=1}^n a_{ij} \x_j(t).
$ \cite{xiao2004fast} showed that for doubly stochastic graphs with spectral gap smaller than one, the weight matrix $A$ converges in linear iterations to the average matrix; thus, $ X(t+1) = A^t X(1)$ asymptotically converges to $\frac{\one\one^T}{n}X(1),$ which guarantees convergence to the initial mean with linear rate. The condition on $A$ being column-stochastic guarantees that average of workers is preserved in each iteration, i.e., $
\bar{X}(t) =\bar{X}(t-1) = \cdots = \bar{X}(1)
$. On the other hand, if the weight matrix $A$ is not row-stochastic,
$\x_i(t)$ converges to $\phi_i\bar{X}(1)$, where $\phi_i$ is the $i$th entry of the stochastic vector $\bm\phi$ with the property that $A^t\rightarrow {\bm\phi}\one^T$. The main approach to tackle consensus in directed networks or for non-doubly stochastic weight matrices is the push-sum protocol introduced for the first time in \cite{kempe2003gossip}. In the push-sum algorithm each worker $i$ updates its auxiliary scalar variable $y_i(t)$ according to the following rule: 
\begin{align*}
    y_i(t+1) = \sum_{j\in\,\mathcal{N}^{\rm{in}}_i} a_{ij}y_j(t).
\end{align*}
Note that the matrix $A$ is column stochastic but not necessarily row-stochastic, thus one can see that  if the scalars $y_i$ are initialized with $y_i(1) = 1$, for all $i \in [n]$, then 
$
    Y(t) = A^t Y(1) = A^t \one. 
$
This implies that $
    Y(t) \xrightarrow{t\rightarrow \infty} n\cdot\bm{\phi}.$
Thus for all $i\in[n]$, we have $
    \frac{\x_i(t)}{y_i(t)} \xrightarrow{t\rightarrow\infty} \frac{\phi_i\one^T\cdot X(1)}{n\cdot{\phi_i}} = \bar{X}(1).$
This shows the asymptotic convergence of $\frac{\x_i}{y_i}$ to the initial mean. Since the parameters $\x_i(t)$ and $y_i(t)$ are kept locally at each node, in every iteration each node can obtain its variable $\z_i(t):=\frac{\x_i(t)}{y_i(t)}$ in the decentralized setting.  
Based on this approach, we present a communication-efficient algorithm for Gossip over directed networks which uses quantization for reducing communication time (Section~\ref{sec:goss_qn}). Moreover, we will show exact convergence  with the same rate as the push-sum algorithm without quantization.

\subsection{Extension to decentralized optimization}
As studied in \cite{tsianos2012push,nedic2014distributed} the push-sum method for reaching consensus among nodes can be extended to decentralized convex optimization problems with some modifications. The push-sum algorithm for decentralized optimization with exact communication, can be summarized in the following steps: 
\begin{align*}
\left\{
      \begin{array}{ll}
      \x_i(t+1) &= \sum_{j=1}^n a_{ij}\, \x_j(t) - \alpha(t)\nabla f_i\left(\z_i(t)\right)\\
      y_i(t+1) &= \sum_{j\in \,\mathcal{N}_i^{\rm{in}}} a_{ij}\, y_j(t)\\
      \z_i(t+1) &= \frac{\x_i(t+1)}{y_i(t+1)}
      \end{array}
      \right.
\end{align*}
Here, local gradients $\nabla f_i$ are computed at the scaled parameters $\z_i(t)$, while the parameters $\z_i(t)$ are obtained similar to the gossip push-sum method. It is shown by \cite{nedic2014distributed} that for all nodes $i\in[n]$ and all $T\ge 1$ the iterates of gossip push-sum method satisfy 
\begin{align*}
f\left(\widetilde{\z}_i(T)\right) -  f\left(\z^\star\right) \le O\left(\frac{1}{\sqrt {T}}\right), 
\end{align*}
for $\alpha(t) = O(1/t)$ and  $\widetilde{\z}_i(T)$ denoting the weighted time average of $\z_i(t)$ for $t=1,\cdots,T$.  This result indicates that for column stochastic matrices, the push-sum protocol achieves the optimal solution at a sublinear rate of $O({1}/{\sqrt {T}})$. In the following section, we show that one can obtain the similar complexity bound even for the case that nodes exchange quantized signals.

\section{Quantized Push-sum for Directed Graphs}\label{sec:algorithm}
In this section, we propose two variants of the push-sum method with \textit{quantization} for both gossip (consensus) and decentralized optimization over directed graphs.

\begin{algorithm}[t]
   \caption{Quantized Push-sum for Consensus over Directed Graphs}
   \label{alg:Gossip}
   \begin{algorithmic}
   \FOR {each node $i\in[n]$ and iteration $t\in[T]$}
   \STATE  $Q_i(t) = Q\left(\x_i(t)-\xh_i(t)\right)$
   \FOR {all nodes $k\in\mathcal{N}^{\rm{out}}_{i}$ and  $j \in\mathcal{N}^{\rm{in}}_i$ }
   \STATE send $Q_i(t)$ and $y_i(t)$ to $k$ and receive $Q_j(t)$ and $y_j(t)$ from $j$.
    \STATE  $\xh_j(t+1) = \xh_j(t) + Q_j(t)$
   \ENDFOR  
     \STATE $\x_i(t+1) = \x_i(t) - \xh_i(t+1) + \sum_{j\in\mathcal{N}^{\rm{in}}_i} a_{ij}\xh_j(t+1)$
   \STATE  $y_i(t+1) = \sum_{j\in\mathcal{N}_i^{\rm{in}}} a_{ij}y_j(t)$
   \STATE $\z_i(t+1) = \frac{\x_i(t+1)}{y_i(t+1)}$
   \ENDFOR
\end{algorithmic}
\end{algorithm}

\subsection{Quantized push-sum for Gossip}\label{sec:goss_qn}
We present a quantized gossip algorithm for the consensus problem in which nodes communicate quantized parameters over a directed graph. As previously noted, our method is derived based on the quantized decentralized optimization method \cite{koloskova2019decentralized} and the push-sum algorithm \cite{tsianos2012push,nedic2014distributed}.The steps of our proposed algorithm are described in Algorithm \ref{alg:Gossip}. Basically, Algorithm \ref{alg:Gossip} consists of two parts: First, the ``Quantization'' step, in which each node computes 
\begin{align*}
    Q_i(t):= Q\left(\x_i(t)-\xh_i(t)\right),
\end{align*}
where $\xh_i(t)$ is an auxiliary parameter stored locally at each node and is being updated at each iteration. Importantly every node $i$, communicates $Q_i(t)$ to its out-neighbors. Quantizing and communicating $\x_i(t) - \xh_i(t)$ instead of $\x_i(t)$ is a crucial part of the algorithm as it guarantees that the quantization noise asymptotically vanishes. Second part of the proposed algorithm is the ``Averaging'' step, in which every node $i$ updates in parallel its parameters $\left(\x_i(t),\,y_i(t),\,\z_i(t)\right)$. The variables $\z_i(t)$ and $y_i(t)$ are updated similar to the push-sum algorithm. For updating $\x_i$, the algorithm uses estimates of the value of $\x_j(t)$ of its in-neighbors, denoted by $\xh_j$. Each worker keeps track of the auxiliary parameters of its in-neighbors $\xh_j(t),\,\text{for all}\,j\in\mathcal{N}^{\rm {\,in}}_i$ with updating it with $Q_j(t)$ received from them:
\begin{align*}
\xh_j(t+1) = \xh_j(t) + Q_j(t).
\end{align*}
Using the same initialization for all $\xh_j(t)$ kept locally in all out-neighbors of node $j$, one can see that $\xh_j(t)$ remains the same among all out-neighbors of node $j$ for all iterations $t$. Similar to the push-sum protocol with exact quantization, the role of $y_i(t)$ in Algorithm \ref{alg:Gossip} is to scale the parameters $\x_i(t)$ of all nodes $i$ in order to obtain $\z_i(t)$ which is the estimation of the average of initial values of nodes $\bar{X}(1)$.


\subsection{Quantized push-sum for decentralized optimization}\label{sec:opt_qn}
Using the push-sum method for optimization problems we propose Algorithm \ref{alg:Opt} for communication efficient collaborative optimization over Directed Graphs. Similar to the method described in Algorithm \ref{alg:Gossip}, this method also has the Quantization and Averaging steps, with the difference that the update rule for $\x_i(t)$ is similar to one iteration of the stochastic gradient descent: 
\begin{align*}
    \x_i(t+1) &= \x_i(t) - \xh_i(t+1) + \sum_{j\in\mathcal{N}^{\rm{in}}_i} a_{ij}\xh_j(t+1) \\- &\alpha\nabla F_i(\z_i(t+1),\zeta_{i,t+1}).
\end{align*}
Importantly, we note that stochastic gradients are evaluated at the scaled values $\z_i(t) = \w_i(t)/ y_i(t)$. One can observe that similar to Algorithm \ref{alg:Gossip}, here the role of $\z_i(t)$ is to \emph{correct} the parameters $\x_i(t)$ through scaling with $y_i(t)$.  

\begin{algorithm}[t]
   \caption{Quantized Decentralized SGD over Directed Graphs}
   \label{alg:Opt}
   \begin{algorithmic}
   \FOR {each node $i\in[n]$ and iteration $t\in[T]$}
   \STATE  $Q_i(t) = Q\left(\x_i(t)-\xh_i(t)\right)$
   \FOR {all nodes $k\in\mathcal{N}^{\rm{out}}_{i}$ and  $j \in\mathcal{N}^{\rm{in}}_i$ }
   \STATE send $Q_i(t)$ and $y_i(t)$ to $k$ and receive $Q_j(t)$ and $y_j(t)$ from $j$.
    \STATE  $\xh_j(t+1) = \xh_j(t) + Q_j(t)$
   \ENDFOR  
     \STATE $\w_i(t+1) = \x_i(t) - \xh_i(t+1) + \sum_{j\in\mathcal{N}^{\rm{in}}_i} a_{ij}\xh_j(t+1)$
   \STATE  $y_i(t+1) = \sum_{j\in\mathcal{N}_i^{\rm{in}}} a_{ij}y_j(t)$
   \STATE $\z_i(t+1) = \frac{\w_i(t+1)}{y_i(t+1)}$
   \STATE $\x_i(t+1) = \w_i(t+1) - \alpha\nabla F_i(\z_i(t+1),\zeta_{i,t+1})$
   \ENDFOR
\end{algorithmic}
\end{algorithm}
In Section \ref{sec:analysis}, we show that the locally kept parameters $\z_i(t)$ will reach consensus at the rate $O(\frac{1}{\sqrt{T}})$ for $\alpha = O(\frac{1}{\sqrt{T}})$. Furthermore, with the same step size, the time average of the parameters $\z_i(t)$ for $t=1,\cdots,T$ will converge to the optimal solution for convex losses and it will converge to a stationary point for non-convex losses with the same rate as Decentralized Stochastic Gradient Descent (DSGD) with exact communication. This reveals that quantization and structure of the graph (e.g, directed or undirected) have no effect on the rate of convergence under the proposed algorithm, however, these dependencies on quantization and the structure of graph appear in the terms of the upper bound.

\section{Convergence Analysis}\label{sec:analysis}
In this section, we study convergence properties of our proposed schemes for quantized gossip and decentralized stochastic learning with convex and non-convex objectives. All proofs are deferred to the Appendix. To do so, we first assume the following conditions on the quantization scheme and loss functions are satisfied.

\begin{ass}[Quantization Scheme]\label{assumption:qsgd}
The quantization function $Q:\mathbb{R}^d\rightarrow\mathbb{R}^d$ satisfies for all $ \boldsymbol{x} \in \mathbb{R}^d$ :
\begin{align}\label{eq:quant}
\E\left[\Big\|Q(\boldsymbol{x})-\boldsymbol{x}\Big\|^2\right]\le \omega^2\left\|\boldsymbol{x}\right\|^2,
\end{align}
where $0\le\omega<1$.
\end{ass}
In the following, we mention a specific quantization scheme and formally state its parameter $\omega$.

\begin{ex}[Low-precision Quantizer]\label{ex:qsgd}\cite{alistarh2017qsgd} 
The unbiased stochastic quantization assigns $\xi_i(\boldsymbol{x},s)$  to each entry $x_i$ in $\boldsymbol{x}$, where $s$ is the number of levels used for encoding $x_i$ and 
\begin{equation*}
\xi_{i}(\boldsymbol{x}, s)=\left\{\begin{array}{ll}
 {(\ell+1) / s} & {\text { w.\,p. \; $\frac{\left|x_{i}\right|}{\|\boldsymbol{x}\|_{2}}s-\ell$ }},\\ 
 {\ell / s} & {otherwise. }
\end{array}\right.
\end{equation*}
Here $\ell$ is the integer satisfying  $0\le\ell<s$ and $
\frac{\left|x_{i}\right|} {\|\boldsymbol{x}\|_{2}} \in[\ell / s,(\ell+1) / s]$. The node at the receiving end, recovers the message according to :
\begin{equation*}
Q\left(x_{i}\right)=\|\boldsymbol{x}\|_{2} \cdot \operatorname{sign}\left(x_{i}\right) \cdot \xi_{i}(\boldsymbol{x}, s).
\end{equation*}
This quantization scheme satisfies Assumption \ref{assumption:qsgd} with \begin{equation*}
\omega^2 = \min \left(d / s^{2}, \sqrt{d} / s\right).
\end{equation*}
\end{ex}
For convenience we assume that the parameters $\x_i$ and $\xh_i$ of all nodes are initialized with zero vectors. This assumption is without loss of generality and is taken to simplify the resulting convergence bounds. 

\begin{ass}[Initialization] \label{assumption:initialization}
The parameters $\x_i$ and $\xh_i$ are initialized with $\mathbf{0}_d$ and $y_i(1) =1$ for all $i \in [n]$. 
\end{ass}
Additionally we make the following assumptions on the local objective function of each computing node and its stochastic gradient.
\begin{ass}[Lipschitz Local Gradients]\label{assumption:lipshitz}
Local functions $f_i(\cdot)$, have $L$-Lipschitz  gradients i.e., for all $i\in[n]$
$$\Big\|\nabla f_i(\y)-\nabla f_i(\x)\Big\|\le L\Big\|\y-\x\Big\|,\; \forall \x,\y \in \mathbb{R}^d. $$
\end{ass}

\begin{ass}[Bounded Stochastic Gradients]\label{assumption:boundedgrad} Local stochastic gradients $\nabla F_i(\x,\zeta_i)$ have bounded second moment i.e., for all $i \in [n]$
$$\Exp_{\zeta_i\sim\mathcal{D}_i}\Big\|\nabla F_i(\x,\zeta_i)\Big\|^2\le D^2, \;\forall \x \in \mathbb{R}^d.$$
\end{ass}
\begin{ass}[Bounded Variance]\label{assumption:boundedvar}
Local stochastic gradients have bounded variance i.e., for all $i \in [n]$ 
$$\Exp_{\zeta_i\sim\mathcal{D}_i}\Big\|\nabla F_i(\x,\zeta_i) - \nabla f_i(\x)\Big\|^2\le \sig^2, \;\forall \x \in \mathbb{R}^d.$$
\end{ass}\par
First, we show that in  Algorithm \ref{alg:Gossip}, the parameters $\z_i$ of all nodes recover the exact value of initial mean in linear iterations. For convenience we denote by $\gamma := \|A-I\|$ and $\widetilde{\la}_1:= \frac{1}{2\la^{-1/2}+4C(\la-\la^{3/2})^{-1}}$.
\begin{thm}[Gossip]\label{thm:gossip}
Recall the definitions of $\lambda$ and $\delta$ in Proposition \ref{lem:propo} and Define $\xi_1=\omega(1+\gamma)\|X(1)\|\cdot\max\{\frac{4 C}{\lambda} ,\,\la^{-1/2}\}$. Under Assumptions \ref{assumption:graph}-\ref{assumption:qsgd},  the iterations of Algorithm \ref{alg:Gossip} satisfy for $\omega \le \frac{\widetilde{\la}_1}{1+\gamma}$ , $i\in[n]$  and all $t\ge1:$  
\begin{align}\label{eq:gossip_thm}
\begin{split}
\mathbb{E}\,\bigg\|\z_{i}(t+1)-\frac{1}{n}\sum_{i=1}^n \x_i(1)\bigg\|&\leq \frac{C\xi_1}{\delta(1-\la^{1/2})} \lambda^{t/ 2} +\\
&\frac{2C\|X(1)\|}{\delta} \lambda^{t}.
\end{split}
\end{align}

\end{thm}

This bound signifies the effect of parameters related to the structure of graph and weight matrix, i.e. $\la,\,C,\,\delta \,\,\text{and} \,\,\gamma$ and the quantization parameter $\omega$. In particular $\omega$ emerges as the coefficient of the larger term, and choosing $\omega = 0$ which corresponds to gossip with exact communication, results in convergence with rate $O(\la^t)$.

Next, we show that the quantization method for decentralized stochastic learning over directed graphs as described in Algorithm \ref{alg:Opt} converges to the optimal solution with optimal rates for convex objectives. In particular we show that global objective function evaluated at time average of $\z_i(t)$ converges to the optimal solution after $T$ iterations with the rate $O(\frac{1}{\sqrt{nT}})$.  The next theorem characterizes the convergence bound for Algorithm \ref{alg:Opt} with convex objectives.
For compactness we define constants $\widetilde{\la}_2 := (1+\frac{6C^2}{(1-\la)^2})^{-1/2}$ and $\xi := 6n D^2(1+\gamma^2)(1+\frac{6C^2}{(1-\la)^2}).$
\begin{thm}[Convex Objectives]\label{thm:convex}
Assume local functions $f_i(\cdot)$ for all $i\in[n]$ are convex, then under Assumptions \ref{assumption:graph}-\ref{assumption:boundedvar}  Algorithm \ref{alg:Opt} for $\omega \le \frac{\widetilde{\la}_2}{\sqrt{6(1+\gamma^2)}}$, $\al = \frac{\sqrt{n}}{8L\sqrt{T}}$ and $T\ge1$ satisfies for all $i\in[n]$ $:$ 
\begin{align}\label{eq:convex_thm} 
\begin{split}
&\E f\left(\frac{1}{T}\sum_{t=1}^T\z_i(t+1)\right) - f(\z^\star)\le
 \\ &  \frac{8L(L+1)}{\sqrt{nT}}\left \|  \z^\star \right\|^2 + 
 \frac{\sig^2(L+1)}{4\,L\sqrt{nT}}+ \\ &\frac{nC^2  \left(L+1\right)\left(\frac{L\sqrt{n}}{2\sqrt{T}}+L+1\right) \left(\xi\omega^2+2nD^2\right)}{10\,\delta^2(1-\la)^2L^2 \,T }.
\end{split}
\end{align}
\end{thm}
\begin{remark}\label{rem:convex}
In the proof of the theorem we show that (see \eqref{eq:convexfinal}) for fixed arbitrary $\alpha$ the error decays with the rate $O(\frac{1}{\al T}) + O(\al^2n) + O(\frac{\al}{n})$. The inequality in the statement of the theorem follows by the specified choice of $\alpha$. More importantly, we highlight that the largest term  in \eqref{eq:convex_thm}, i.e., $\frac{1}{\sqrt{T}}$, is directly proportionate to $\frac{1}{\sqrt{n}}$ and $\sig^2$ which emphasizes the impact of the number of workers and mini-batch size in accelerating the speed of convergence. Additionally parameters related to the structure of graph, i.e., $\la, C, $ and $\delta$ and the quantization parameter $\omega$, only appear in the terms of order $\frac{1}{T}$ and $\frac{1}{T\sqrt{T}}$ which are asymptotically negligible compared to $\frac{1}{\sqrt{T}}$.
\end{remark}
In the next theorem we show convergence of Algorithm \ref{alg:Opt} with non-convex objectives. Importantly, we demonstrate that the gradient of global function converges to the zero vector with the same optimal rate as in decentralized optimization with exact-communication over undirected graphs(See \cite{lian2017can}).
%
%
\begin{thm}[Non-convex Objectives]\label{thm:nonconvex}
Under Assumptions \ref{assumption:graph}-\ref{assumption:boundedvar}, Algorithm \ref{alg:Opt} after sufficiently large iterations $(T\ge4n)$, and for $\omega \le \frac{\widetilde{\la}_2}{\sqrt{6(1+\gamma^2)}}$ and $\al = \frac{\sqrt{n}}{L\sqrt{T}}$ satisfies:
\begin{align}\label{eq:nonconvex_thm}
\begin{split}
&\frac{1}{T}\sum_{t=1}^T \,\mathbb{E}\left\|\nabla f\left(\frac{1}{n}\sum_{i=1}^n \x_i(t)\right)\right\|^{2}+\\ &\frac{1}{2T}\sum_{t=1}^{T} \,\mathbb{E}\left\|\frac{1}{n}\sum_{i=1}^n \nabla f_i(\z_i(t+1))\right\|^{2} \le \frac{\sig^2L}{\sqrt{nT}} + \\
&\frac{2L\left(f\left(0\right)-f^\star\right)}{\sqrt{nT}}+ \frac{12C^2}{\delta^2(1-\la)^2T}\left(\xi \omega^2+2nD^2\right).
\end{split}
\end{align}

Moreover for all $i \in[n]$ and $T\ge1$, it holds that
\begin{align}\label{eq:consen_thm}
\begin{split}
\mathbb{E}\bigg\|\z_{i}(T+1)-&\frac{1}{n}\sum_{i=1}^n \x_i(T)\bigg\|^{2} \le\\& \frac{6C^2n}{\delta^2(1-\la)^2L^2T}\left(\xi\omega^2+2nD^2\right).
\end{split}
\end{align}
\end{thm}
\begin{remark}
The inequality \eqref{eq:nonconvex_thm} implies convergence of the average of $\x_i(t)$ among workers to a stationary point as well as convergence of average of local gradients $\nabla f_i(\z_i(t))$ to zero with the optimal rate $O(\frac{1}{\sqrt{T}})$ for non-convex objectives. Interestingly similar to the convex-case discussed in Remark~\ref{rem:convex}, the number of workers $n$ and  stochastic gradient variance $\sig^2$ emerge in the dominant terms while the parameters related to weight matrix, graph structure and quantization appear in the term of order $O(\frac{1}{T})$. 
\end{remark}
\begin{remark}
As the proof of theorem shows, for arbitrary fixed step size $\alpha$ we derive the inequality in \eqref{eq:nonconvex_fixedal} and the desired result of theorem is concluded by the specified choice of $\alpha$ in the theorem. Importantly we note that the condition on the number of iterations ,i.e. $T\ge4n$, is a direct result of the specified choice of $\alpha$. Therefore one can get the same rate for convergence for all $T\ge1$, with other choices for the step size. For example using the relation \eqref{eq:nonconvex_fixedal}, by choosing $\alpha = \frac{L}{2\sqrt{T}}$ we achieve convergence for all $T\ge1.$
\end{remark}
\begin{remark}
Based on \eqref{eq:consen_thm}, the parameters $\z_i$ of nodes reach consensus with the rate $O(\frac{1}{T})$ for the specified value of $\alpha$. For an arbitrary value of $\alpha$, consensus is achieved with the rate $O(\alpha^2)$ (see \eqref{eq:cons_error}) which implies that smaller values of $\alpha$ will result in faster consensus among the nodes. However due to the term $O(\frac{1}{\al T})$ in the convergence of objective to the optimal solution or in the convergence to a stationary point, this fast consensus will be at the cost of slower convergence of objective function in both convex and non-convex cases.
\end{remark}
\section{Numerical Experiments}\label{sec:numerical}
In this section, we compare the proposed methods for communication-efficient message passing over directed graphs, with the push-sum protocol using exact communication (e.g., as formulated in  \cite{kempe2003gossip,tsianos2012push} for gossip or optimization problems). Throughout the experiments we use the strongly-connected directed graphs $\mathcal{G}_1$ and $\mathcal{G}_2$ with $10$ vertices as illustrated in Figure \ref{fig:graphs}. For each graph we form the column-stochastic weight matrix according to the method described in Example \ref{ex:graph}. In order to study the effect of graph structure we consider $\mathcal{G}_2$ to be more dense with more connection between nodes. For quantization, we use the method discussed in Example \ref{ex:qsgd} with $B = log(s) +1 $ bits used for each entry of the vector (one bit is allocated for the sign). Moreover the norm of transmitted vector and the scalars $y_i$ are transmitted without quantization. In the push-sum protocol with exact communication the entries of the vector $\x_i$ and the scalar $y_i$ are each transmitted with 54 bits. 

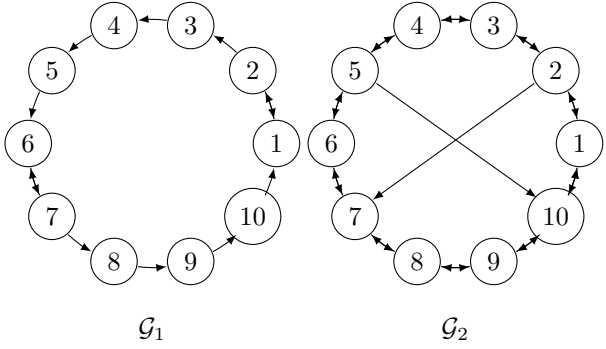
\begin{figure} 
\begin{center}
\centerline{
\begin{tikzpicture}
\def \n {10}
\def \radius {1.65cm}
\def \margin {11.3} 
\foreach \s in {1,...,\n}
{
  \node[draw, circle] at ({360/\n * (\s - 1)}:\radius) {$\s$};
  \draw[->, >=latex] ({360/\n * (\s - 1)+\margin}:\radius) 
    arc ({360/\n * (\s - 1)+\margin}:{360/\n * (\s)-\margin}:\radius);
} 
\foreach \s in {2,7}
    {
    \draw[->, >=latex] ({360/\n * (\s-1 )-\margin}:\radius) 
    arc ({360/\n * (\s-1)-\margin}:{360/\n * (\s-2)+\margin}:\radius);
    }
    \node[below] at (0,-2.2) {$\mathcal{G}_1$};
    \end{tikzpicture}\;\begin{tikzpicture}
\def \n {10}
\def \radius {1.65cm}
\def \margin {11.3} 
\foreach \s in {1,...,10}
{
  \node[draw, circle] at ({360/\n * (\s - 1)}:\radius) {$\s$};
  \draw[->, >=latex] ({360/\n * (\s - 1)+\margin}:\radius) 
    arc ({360/\n * (\s - 1)+\margin}:{360/\n * (\s)-\margin}:\radius);
}
  \draw[->, >=latex] ({360/\n * (9)+\margin}:\radius) 
    arc ({360/\n * (9)+\margin}:{360/\n * (10)-\margin}:\radius);
\foreach \s in {1,...,\n}
    {
    \draw[->, >=latex] ({360/\n * (\s-1 )-\margin}:\radius) 
    arc ({360/\n * (\s-1)-\margin}:{360/\n * (\s-2)+\margin}:\radius);
    }
    \draw[->, >=latex] (1.05,0.8)--(-1.1,-0.75) ;
    \draw[->, >=latex] (-1.05,0.8)--(1.08,-0.72) ;
    \node[below] at (0,-2.2) {$\mathcal{G}_2$};
    \end{tikzpicture}}
    \caption{The experimented directed graphs representing communication between computing nodes}
    \label{fig:graphs}
    \end{center}
\end{figure} 

\noindent \textbf{Gossip experiments.}
First, we evaluate performance of Algorithm \ref{thm:gossip} for gossip type problems. We initialize the parameters $\x_i(1) \in [0,1]^{1024}$ of all nodes to be $i.i.d.$ uniformly distributed random variables. Moreover we initialize the auxiliary parameters $\xh_i = \mathbf{0}$ and $y_i=1$ for all $i\in[n]$. In Figure \ref{fig:gossip}(Top) we compare the performance of Algorithm \ref{alg:Gossip} with the push-sum protocol with exact-communication for both networks $\mathcal{G}_1$ and $\mathcal{G}_2$. While Algorithm \ref{alg:Gossip} has almost the same performance as push-sum over $\mathcal{G}_1$, it is outperformed by push-sum over $\mathcal{G}_2$.

\begin{figure}[t] 
\centering 
\subfloat{%
  \includegraphics[clip,width=0.78\columnwidth,height=0.5\columnwidth]{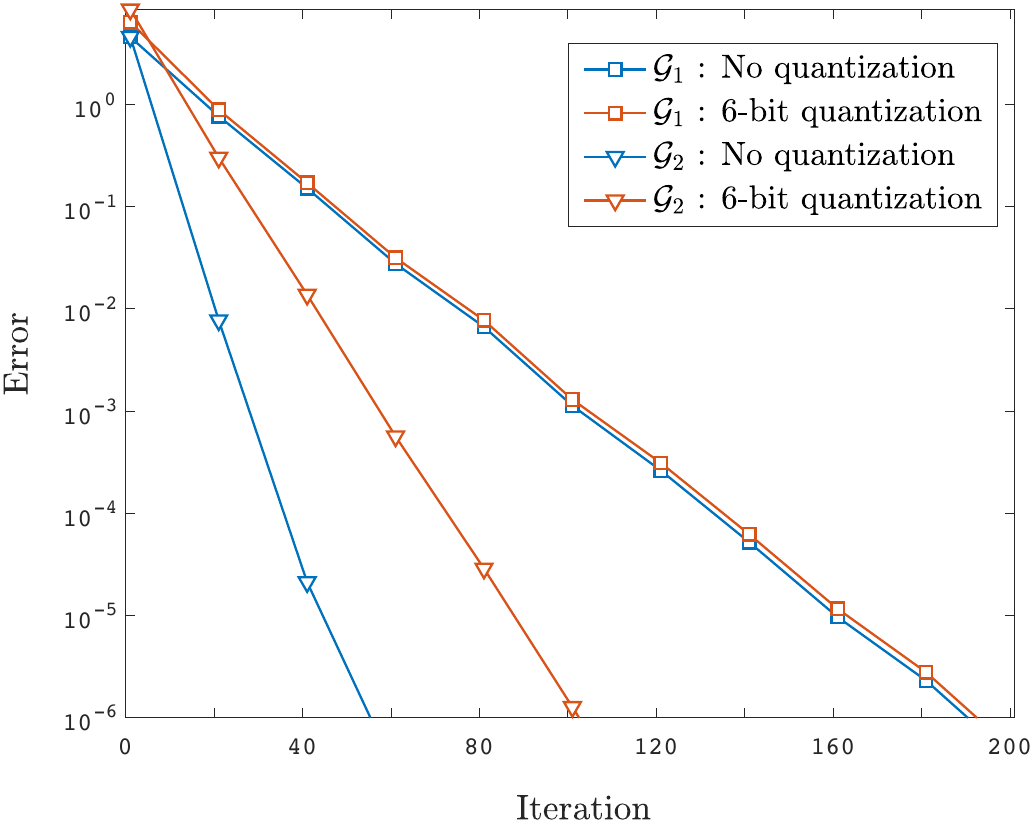}%
}

\subfloat{%
  \includegraphics[clip,width=0.78\columnwidth,height=0.5\columnwidth]{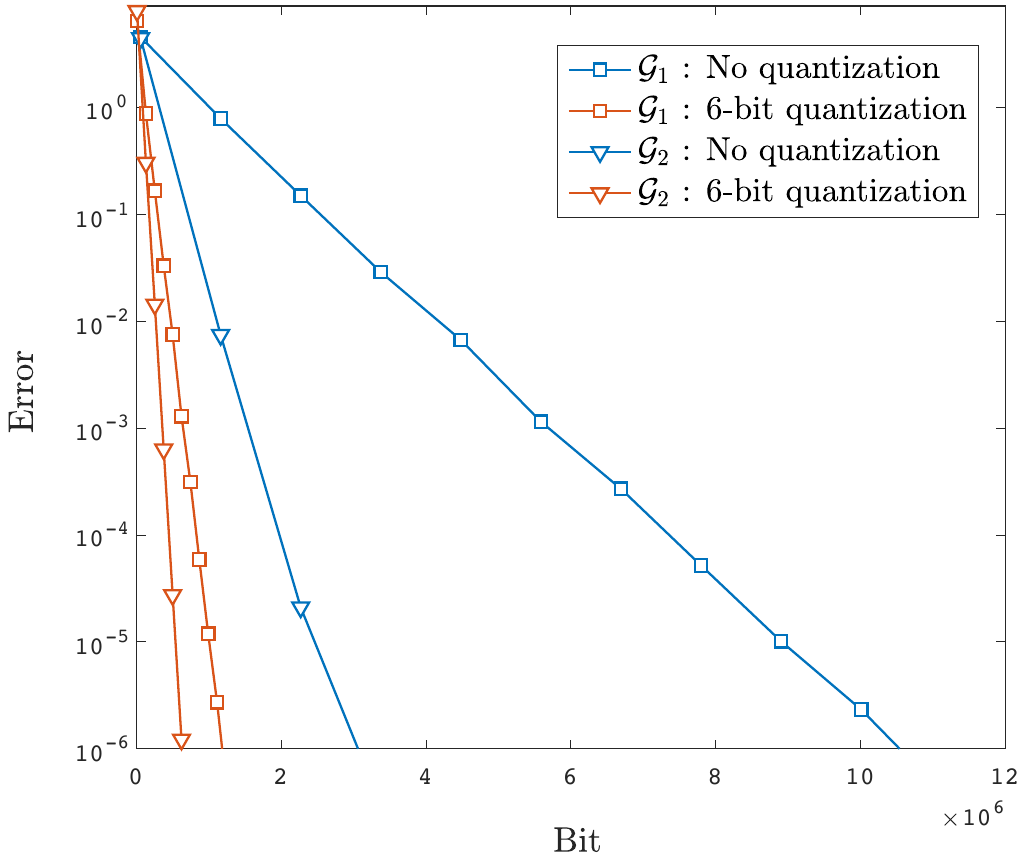}%
}
\caption{Comparison of the proposed algorithm for the gossip problem and the push-sum protocol using exact-communication based on iteration (Top) and total number of bits communicated between two neighbor nodes (Bottom) over the graphs $\mathcal{G}_1$ and $\mathcal{G}_2$.}
    \label{fig:gossip}
    \end{figure}
 
However the superiority of exact-communication methods compared in each iteration could be predicted. In order to compare the two methods based on time spent to reach a specific level of error, we compare their performances based on the number of bits that each worker communicates.
In Figure \ref{fig:gossip} (Bottom) the number of bits required to reach a certain error performance is illustrated for both methods. For the graphs $\mathcal{G}_1$ and $\mathcal{G}_2$ we observe up to 10x and 6x reduction in the total number of bits, respectively. 

\noindent \textbf{Decentralized optimization experiments.} Next, we study the performance of Algorithm \ref{alg:Opt} for decentralized stochastic optimization using convex and non-convex objective functions. First, we consider the objective
\begin{align*}
    f(\x) = \frac{1}{2nm}\sum_{i=1}^n\sum_{j=1}^m \left\|\x - \zeta_j^i\right\|^2,
\end{align*}
where we set $m=n=10$ and $d = 256$. Thus each node $i$ has access to its local data-set $\{\zeta_1^i,\zeta_2^i,\cdots,\zeta_{10}^i\}$ and is using one sample at random in  each iteration to obtain the stochastic gradient of its own local function. Here we use the data-set generated according to 
\begin{align*}
    \zeta_j^i \overset{i.i.d.}{\sim} \zeta^\star + \mathcal{N}\left(\mathbf{0},\mathbf{I}_{\,256}\right), \quad\text{for all}\; i,j,
\end{align*}
where $\zeta^\star$ is a fixed vector initialized as $\mathrm{Uniform} [0,100]^{256}$. The step size $\alpha$ for each setting, is fine-tuned up to iteration $50$ among $20$ values in the interval $[0.01,\, 3]$. The Loss at iteration $T$ is presented by $\frac{1}{d}\|\widetilde{\z}_1(T) - \zeta_{\,\rm opt} \|$, where $\widetilde{\z}_1(T) := \frac{1}{T}\sum_{t=1}^T\widetilde{\z}_1(t)$ is the time-average of the model of worker $1$ and $\zeta_{\rm{\,opt}}$ is the optimal solution.

\begin{figure} 
\centering
\subfloat{%
  \includegraphics[clip,width=0.78\columnwidth,height=0.52\columnwidth]{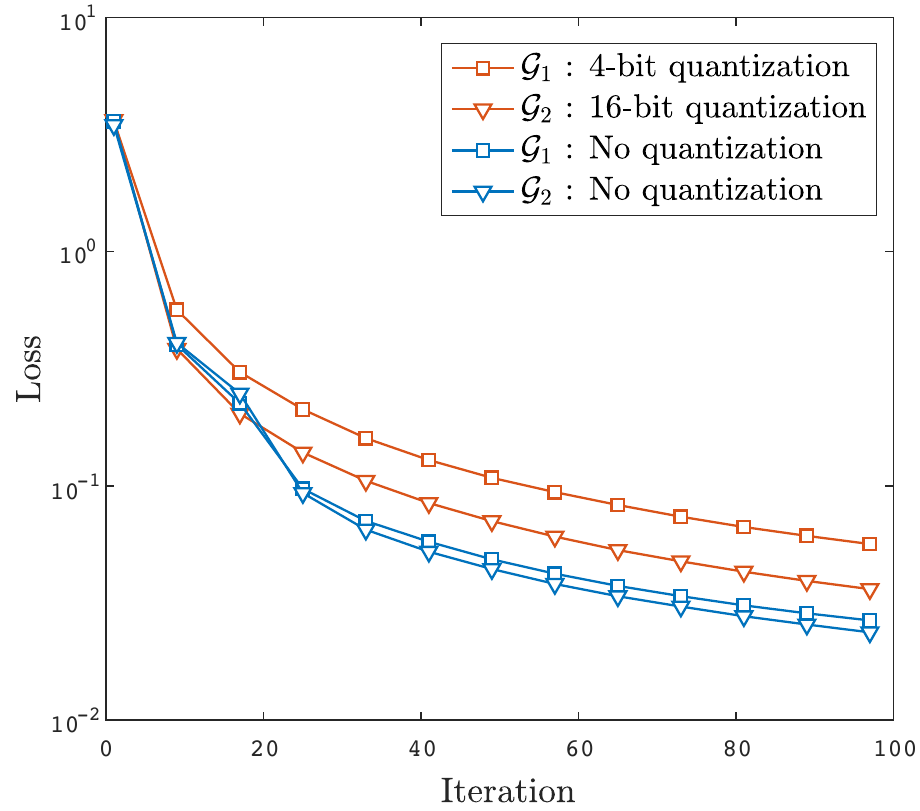}%
}

\subfloat{%
  \includegraphics[clip,width=0.78\columnwidth,height=0.52\columnwidth]{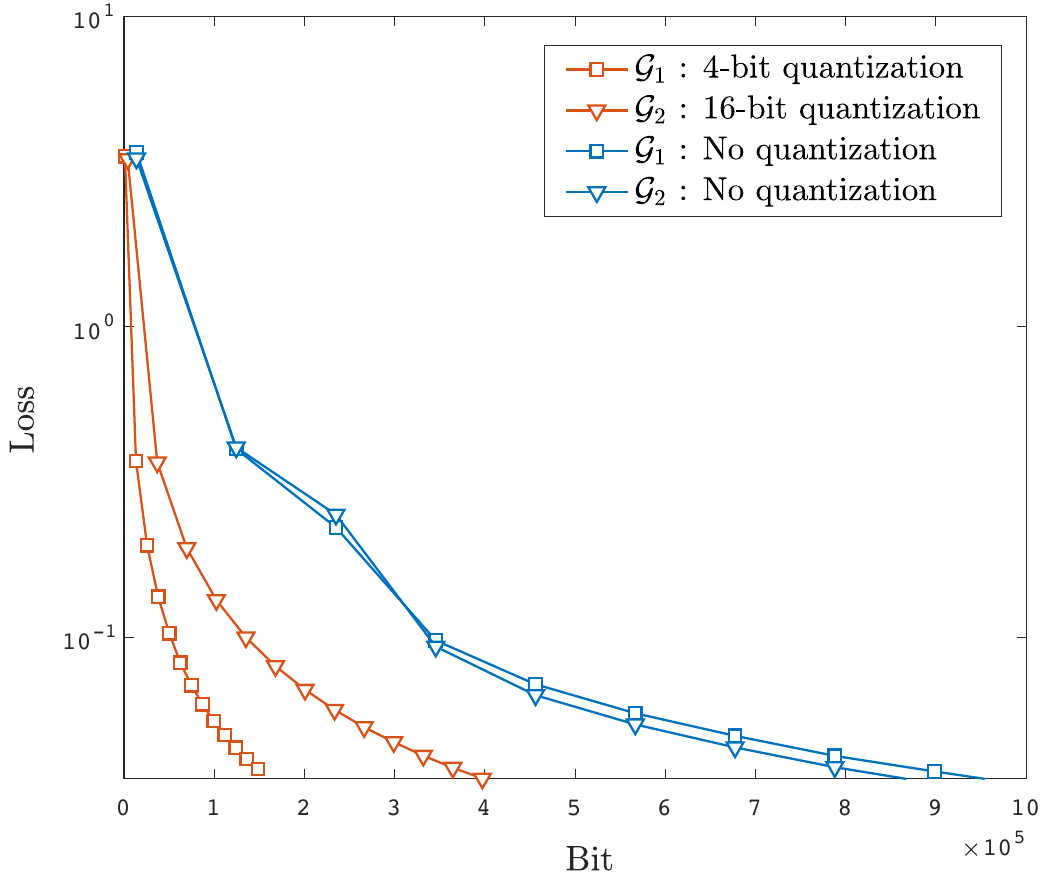}%
}
    \caption{Comparison of the proposed method and exact-communication push-sum method using least-square as objective, based on the iteration number (Top) and total number of communicated bits over the graphs $\mathcal{G}_1$ and $\mathcal{G}_2$.}
    \label{fig:convex}
    \end{figure}

The results of this experiment are in Figure \ref{fig:convex} (Top) which illustrates the convergence of Algorithm \ref{alg:Opt} based on the number of iterations for different levels of quantization and over the two graphs $\mathcal{G}_1$ and $\mathcal{G}_2$. The non-quantized method outperforms the quantized methods based on iteration. This is due to the quantization noise injected in the flow of information over the graph which depends on the number of bits each node uses for encoding and the structure of graph. However this error asymptotically vanishes resulting in small overall quantization noise. This implies that with less quantization noise (i.e., using more bits to encode) the loss decay based on iteration number gets smaller. However as we observe in Figure~\ref{fig:convex} (Bottom), more quantization levels will result in larger number of bits required to achieve a certain level of loss. Consequently, the push-sum protocol with exact communication for optimization over directed graphs is not communication-efficient as we demonstrated that using smaller number of bits with Algorithm \ref{alg:Opt} results in $5$x reduction in transmitted bits. 

As we showed in Theorem \ref{thm:nonconvex} in Section \ref{sec:analysis}, our proposed method guarantees convergence to a stationary point for non-convex and smooth objective functions. In order to illustrate this, we train a neural-network with $10$ hidden units with sigmoid activation function to classify the MNIST data-set into $10$ classes. We use the graph $\mathcal{G}_1$ with $10$ nodes where each node has access to $1000$ samples of data-set and uses a randomly selected mini-batch of size $10$ for computing the local stochastic gradient descent. For each setting, the step-size $\alpha$ is fine-tuned up to iteration $200$ and over $15$ values in the interval $[0.1,\, 3]$. Figure \ref{fig:nonconvex} illustrates the results for training loss of two aforementioned methods based on number of iteration (Top) and total number of bits communicated between two neighbor nodes (Bottom). We note the close gap in each iteration between the loss decay of our proposed method with $8$ bits quantization, and the push-sum with exact communication. However since our method uses significantly less bits in each iteration, it reaches the same training loss in fewer iterations. In particular Figure \ref{fig:nonconvex} (Bottom) demonstrates $5$x reduction in total number of bits communicated using our proposed method.

\begin{figure}
\centering
\subfloat{%
  \includegraphics[clip,width=0.78\columnwidth,height=0.52\columnwidth]{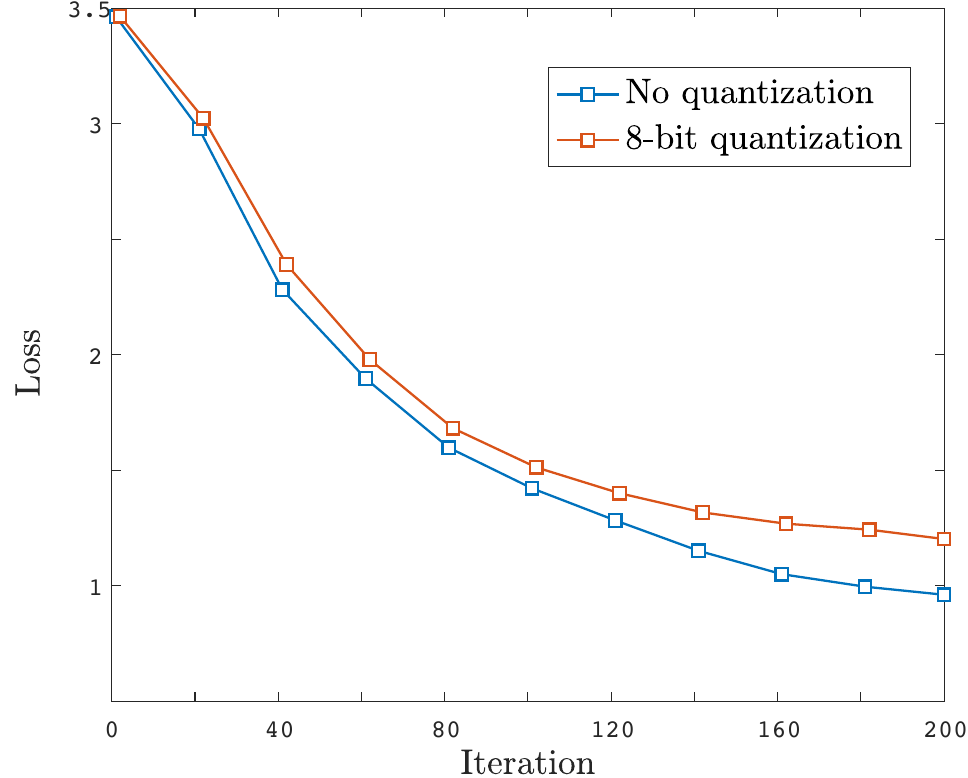}%
}

\subfloat{%
  \includegraphics[clip,width=0.78\columnwidth,height=0.52\columnwidth]{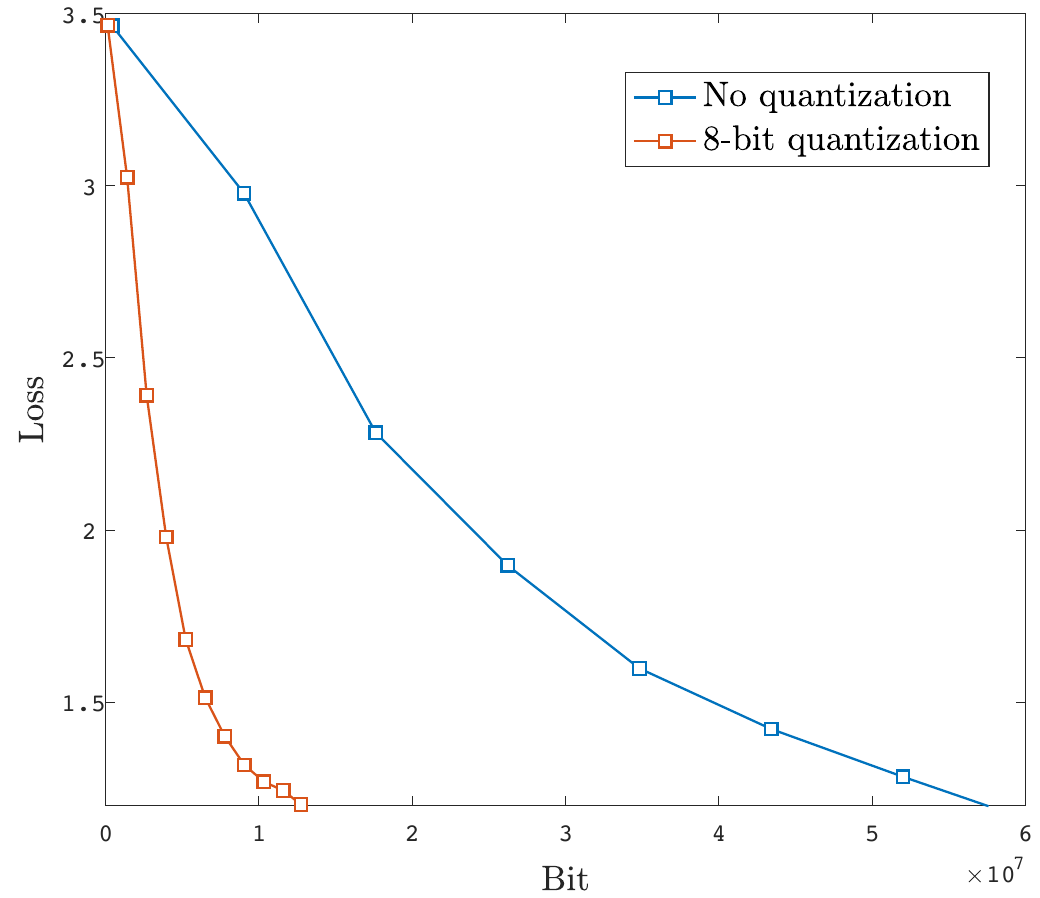}%
}
\caption{Comparison of the proposed method and exact-communication push-sum method in training a neural network with MNIST data-set, based on the iteration number (Top) and total number of communicated bits (Bottom).} 
\label{fig:nonconvex}
\end{figure}

\section{Conclusion and Future Work}
In this paper, we proposed a scheme for communication-efficient decentralized learning over directed graphs. We showed that our method converges at the same convergence rate as non-quantized methods for both gossip and decentralized optimization problems. As we demonstrated in Section \ref{sec:numerical}, the proposed approach results in significant improvements in communication-time of the push-sum protocol. An interesting future direction is extending these results to algorithms that achieve linear convergence for strongly-convex problems (e.g., \cite{xi2017dextra}). Another direction is extending our results to asynchronous decentralized optimization over directed graphs.

\section*{Acknowledgements}
This work was supported by National Science Foundation (NSF) under grant CCF-1909320 and UC Office of President under Grant LFR-18-548175.
\bibliography{C_ICML}
\bibliographystyle{icml2020}
\onecolumn
\section*{\Large Appendix}
In this section, we first provide further numerical experiments as well as presenting the hyper-parameters of all methods used in the experiments. We then continue with presenting proofs of the Theorems \ref{thm:gossip}-\ref{thm:nonconvex}.
\subsection{Neural Network Training on CIFAR-10 dataset}
We train a neural network of 20 hidden-units with sigmoid activation functions for binary classification of the CIFAR-10 dataset. We use the topology of $\mathcal{G}_1$ as in Figure \ref{fig:graphs} with the corresponding weight matrix constructed according to Example \ref{ex:graph}. The value of step size $\alpha$ is fine-tuned among 15 values of $\al$ uniformly selected in $[0.1,3]$ and up to iteration 300 for both algorithms.    
\begin{figure}[H]
    \centering
    \hspace{0mm}\includegraphics[width=7cm,height=5.9cm]{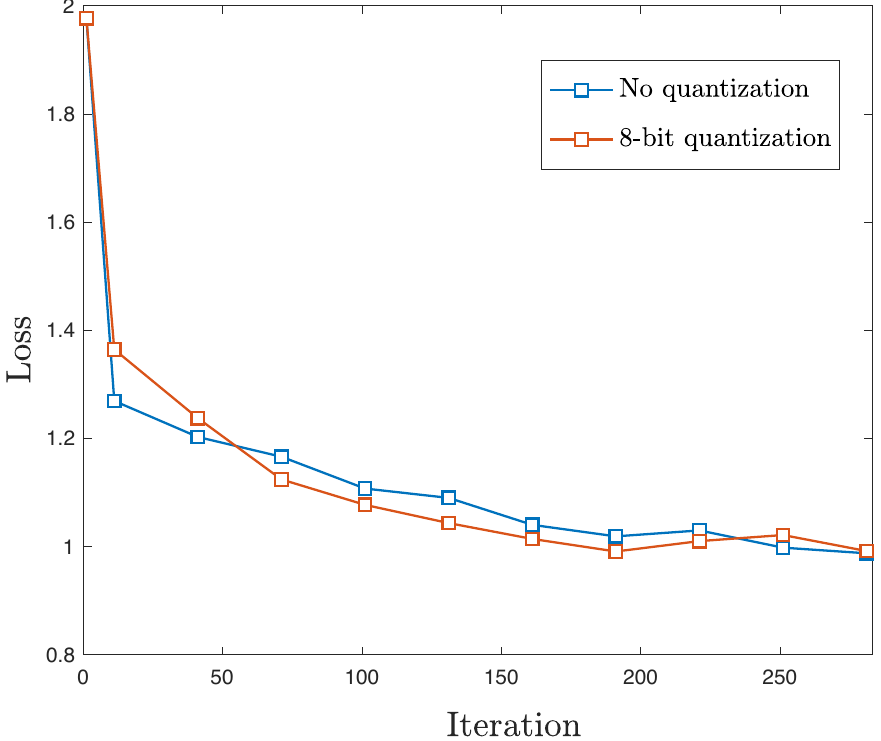}\quad\quad
    \includegraphics[width=7.4cm,height=6cm]{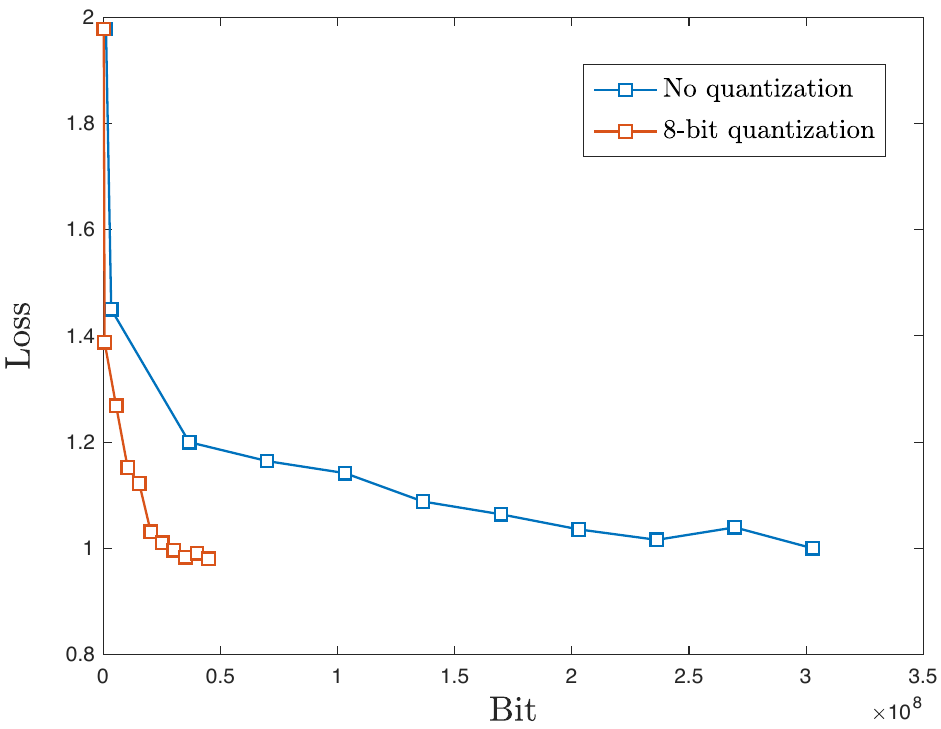}\vspace{0mm}
    \caption{Comparison of the proposed method and exact-communication push-sum method in training a neural network with CIFAR-10 data-set, based on the iteration number (Left) and total number of bits communicated between two neighbor nodes (Right).}
    \label{fig:CIFAR}
    \vspace{-0.1in}
\end{figure}
Figure \ref{fig:CIFAR} illustrates the performance of the proposed algorithm based on iteration number (Left) and number of bits communicated between two neighbor nodes (Right), and compares the vanilla push-sum method and the proposed communication-efficient method with 8 bits for quantization. We use SGD with mini-batch size of 100 samples for each node in both methods. We highlight the close similarity in training loss of the two methods for a fixed number of iteration. More importantly the proposed method uses remarkably smaller number of bits implying faster communication while reaching the same level of training loss.  
\subsection{Details of the Numerical Experiments}
The step-sizes of the algorithms used in Section \ref{sec:numerical} are fine tuned over the interval $[0.01, \,3]$, so that the best error achieved by each method is compared. In Table \ref{tab:param} we present the fine-tuned step sizes as well as model size (i.e. dimension of parameteres) and the mini-batch size of the algorithms that are used throughout the numerical experiments. 
\begin{table}[H]
\caption{Details on the hyper-parameters of the vanilla push-sum algorithm with no quantization and the proposed quantized push-sum algorithm.}
\label{tab:param}
\vskip 0.1in
\begin{center}
\begin{small}
\begin{sc}
\begin{tabular}{lcccccr}
\toprule
Objective & iteration &model size& mini-batch size& graph & step size \\
\midrule
Square loss     &50 & 256&1&$\mathcal{G}_1$ & 1.7  \\
Square loss    &50 & 256 &1&$\mathcal{G}_2$ &  1.1\\
NN, MNIST &200 & 7960&10&$\mathcal{G}_1$&     2.2   \\
NN, CIFAR-10  &300 &20542 &100&$\mathcal{G}_1$& 1.1   \\
Square loss (4-bits)     &50 &256&1&$\mathcal{G}_1$ & 1.1  \\
Square loss (16-bits)     &50 &256&1&$\mathcal{G}_2$ &  1.1\\
NN, MNIST (8-bits) &200 & 7960&10& $\mathcal{G}_1$&     1.9   \\
NN, CIFAR-10 (8-bits)  &300 &20542&100& $\mathcal{G}_1$& 0.3   \\
\bottomrule
\end{tabular}
\end{sc}
\end{small}
\end{center}
\vskip -0.1in
\end{table}

\section*{\Large Proofs}
\subsection*{Notation}
Throughout this section we set the following notation. For variables, stochastic gradients and gradients, we concatenate the row vectors corresponding to each node to form the following matrices:
\begin{align*}
&Z(t) :=\left[{\z_1(t)}\,; \quad \z_2(t)\quad{\cdots}\,\quad  {\z_n(t)}\right] \in \mathbb{R}^{n \times d},\\
&\partial F\left(Z(t), \zeta_{t}\right):=\left[\nabla F_{1}\left(\z_1(t) , \zeta_{1,t}\right)\,; \quad\nabla F_{2}\left(\z_2(t) , \zeta_{2,t}\right)\, \quad \cdots \quad \nabla F_{n}\left(\z_n(t) , {\zeta}_{n, t}\right)\right] \in \mathbb{R}^{n \times d},\\
&\partial f\left(Z(t)\right):=\left[\nabla f_{1}\left(\z_1(t) \right)\,; \quad\nabla f_{2}\left(\z_2(t) \right)\, \quad \cdots \quad \nabla f_{n}\left(\z_n(t) \right)\right] \in \mathbb{R}^{n \times d}.
\end{align*}

\section{Proof of Theorem \ref{thm:gossip} : Quantized Gossip over Directed Graphs}
First we write the iterations of Algorithm \ref{alg:Gossip} in matrix notation to derive the following: 

\begin{align}\label{alg:gossip_matrix}
\left\{
      \begin{array}{ll}
    &Q(t)=Q\left(X(t)-\Xh(t)\right)\\
    &\Xh(t+1)=\Xh(t)+Q(t)\\
   &X(t+1)=X(t)+(A-I) \Xh(t+1)\\
   &\y(t+1)=A \y(t)\\
    &\z_{i}(t+1)=\frac{\x _i(t+1)}{y_i(t+1)}\\
      \end{array}
      \right.
\end{align}

Based on this, we can rewrite the update rule for $X(t+1)$ as follows:

\begin{equation}\nn
X(t+1)=A X(t)+(A-I)(\Xh(t+1)-X(t)).
\end{equation}
By repeating this for $X(t),...,X(1)$, the update rule for $X(t+1)$ takes the following shape:

\begin{equation}\label{eq:finalxt}
X(t+1)=A^{t} X(1)+\sum_{s=0}^{t-1} A(A-I)(X(t-s+1)-X(t-s)).
\end{equation}
Multiplying both sides by $\one^T$ and recalling that $\one^TA=\one^T$, yields that

\begin{align}
\one^TX(t+1) =\one^TX(1).
\end{align}

With this and \eqref{eq:finalxt} we have for all $t\ge1$:
\begin{align}\label{eq:updateXt}
\begin{split}
\bigg\|X(t+1)-\bm{\phi} \one^T X(1)\bigg\|&=\left\|A^{t} X(1)+\sum_{s=0}^{t-1} A^{s}(A-I)(\Xh(t-s+1)-X(t-s))-\bm{\phi} \one^T X(1)\right\|\\
&\leqslant C \lambda^{t}\bigg\|X(1)\bigg\|+\,2\,C \sum_{s=0}^{t-1} \lambda^{s}\bigg\| \Xh(t-s+1)-X(t-s)\bigg\|.
\end{split}
\end{align}
Furthermore by the iterations of Algorithm \ref{alg:gossip_matrix} as well as the assumption on quantization noise in Assumption \ref{assumption:qsgd} we find that
\begin{align*}
\begin{split}
\E\bigg\|X(t+1)-\Xh(t+2)\bigg\|&=\E\bigg\| X(t+1)-\Xh(t+1)-Q(t+1)\bigg\| \\
&\leq\omega\bigg\|X(t+1)-\Xh(t+1)\bigg\| \\
&=\omega\bigg\|X(t)+(A-I) \Xh(t+1)-\Xh(t+1)\bigg\| \\
&\leq\omega\bigg\|X(t)-\Xh(t+1)\bigg\|+\omega\bigg\|(A-I) \Xh(t+1)\bigg\|. 
\end{split}
\end{align*}
Next, we add and subtract $(A-I)X(t)$ to the RHS and also use the fact that $A\bm \phi = \bm \phi$ ( \cite{zeng2015extrapush}) to conclude that
\begin{align}\label{eq:xt-xht}
\begin{split}
\E\bigg\|X(t+1)-\Xh(t+2)\bigg\|&\le\omega\bigg\|X(t)-\Xh(t+1)\bigg\|+\omega\bigg\|(A-I)(\Xh(t+1)-X(t))\bigg\| \\
&+\omega\bigg\|(A-I)\left(X(t)-\bm{\phi} \one^{T} X(1)\right)\bigg\|.
\end{split}
\end{align}
Let $\gamma:=\|A-I\|, \,\lambda_{1}:=\omega(1+\gamma)$ and $ \lambda_{1}':=\omega \gamma$, then we conclude that 
\begin{align}\nn
\E\bigg\|X(t+1)-\Xh(t+2)\bigg\| =\la_1\bigg\|\Xh(t+1)-X(t)\bigg\|+\la_1' \bigg\| X(t)-\bm{\phi} \one^T X(1) \bigg\|.
\end{align}
Denoting by $R(t) :=\E\bigg\|X(t)-\bm{\phi} \one^T X(1)\bigg\|$ and  $U(t)=\E\bigg\| \Xh(t+1)-X(t)\bigg\|$, we derive the next two inequalities based on \eqref{eq:updateXt} and \eqref{eq:xt-xht}:
\begin{align}\label{eq:UandR}
\left\{\begin{array}{l}
{\mathbb{E}\,R(t+1) \leq C \cdot \lambda^{t}\|X(1)\|\,+2\,C \sum_{s=0}^{t-1} \lambda^{s} U(t-s)} \\[5pt]
{ \mathbb{E}\,U(t+1) \leq \lambda_{1} U(t)+\lambda_{1}' R(t)}
\end{array}\right.
\end{align}
\begin{lem}\label{lem:Uineq}
The iterates of $U(t)$ satisfy for all iterations $t\ge1$ 
\begin{align}\nn
{U(t) \leq  \xi_1\la^{t/2}},
\end{align}
where $\xi_1=\max\{\frac{4 C}{\lambda} \lambda_{1}\|X(1)\|,\,\frac{\la_1\|X(1)\|}{\la^{1/2}}\} $ and for the values of  $\la_1$ chosen such that $\la_1 \le \frac{1}{2}(\frac{1}{\la^{1/2}}+\frac{2C}{\la-\la^{3/2}})^{-1}$\,.
\end{lem}
\begin{proof}
Noting that $\lambda_1' \le \lambda_1$, we have by \eqref{eq:UandR} for all $t\ge1$: 
\begin{align}\label{eq:eqU}
U(t+1) \leq \lambda_{1} U(t)+2C\cdot \lambda_{1}\left(\lambda^{t-1}\|X(1)\|+\sum_{s=0}^{t-2} \lambda^{s} U\left(t-s-1\right)\right).
\end{align}
The proof is based on induction on the inequality in \eqref{eq:eqU}.
Let  $U(t) \leq \xi_1 \cdot \lambda^{t / 2}$, then by \eqref{eq:eqU}
\begin{align*}
U(t+1) &\leq \lambda_{1} \cdot \xi_1 \cdot \lambda^{t / 2} + 2C \lambda_{1}\left(\lambda^{t-1}\|X(1)\|\right)+ 2C \lambda_{1} \cdot \xi_1 \cdot \frac{\lambda^{\frac{t-1}{2}}}{1-\lambda^{1 / 2}} \\
&=\frac{\lambda_{1} \cdot \xi_1}{\lambda^{1 / 2}} \lambda^{\frac{t+1}{2}}+2C \lambda_{1}\left(\|X(1)\| \cdot \lambda^{\frac{t-3}{2}}\right)\lambda^{\frac{t+1}{2}}+\frac{2C \cdot \lambda_{1} \cdot \xi_1}{\lambda-\lambda^{3 / 2}} \lambda^{\frac{t+1}{2}} \\
&\leq \lambda_{1} \xi_1\left(\frac{1}{\lambda^{1 / 2}}+\frac{2C}{\lambda-\lambda^{3 / 2}}\right) \lambda^{\frac{t+1}{2}}+\frac{2C \cdot \lambda_{1}\|X(1)\|}{\lambda} \cdot \lambda^{\frac{t+1}{2}} \\
&\le \left(\frac{1}{2} \xi_1+\frac{2C \cdot \lambda_{1}\|X(1) \|}{\lambda}\right) \cdot \lambda^{\frac{t+1}{2}}\le\xi_1 \lambda^{\frac{t+1}{2}},
\end{align*}
where the last two steps follow by the assumptions of the lemma on $\xi_1$ and $\la_1$. Moreover for $U(1)$ we follow the inequalities similar to \eqref{eq:xt-xht} to find that 
$$
U(1) \le \omega \| X(1) \| \le \xi_1\la^{1/2},
$$
where we used $\omega \le \la_1$ in the last inequality. This completes the proof of the lemma.
\end{proof}
Lemma \ref{lem:Uineq} implies that the quantization error $U(t)$ is decaying with the rate $\la^{t/2}$.  As we will see shortly, this results in total error, decaying with the same rate. From the update rule in Algorithm \ref{alg:gossip_matrix}, we obtain the following for all $t\ge0$:
\begin{align*}
&\y(t+1)=A \y(t)=A^{t} \y(1) .
\end{align*}
Since by assumption $\y(1) = \one$ it yields that
\begin{align*}
&\y(t+1)=A^{t} \one=\left(A^{t}-\bm{\phi} \one^{T}\right) \one+\bm{\phi} \one^{T} \one\\
&=\left(A^{t}-\bm{\phi} \one^{T}\right) \one+\bm{\phi} n .
\end{align*}
Therefore for all $i \in [n]$
\begin{align*}
\y_{i}(t+1)=\left[\left(A^{t}-\bm{\phi} \one^{T}\right)\right]_{i}+\bm{\phi}_{i} n.
\end{align*}
Furthermore, based on Algorithm \ref{alg:gossip_matrix} the parameter $\z_i(t+1)$ satisfies:
\begin{align}\label{eq:z}
\begin{split}
\z_{i}(t+1)=\frac{\x_{i}(t+1)}{y_{i}(t+1)}=& \frac{\left[A^{t} X(1)+\sum_{s=0}^{t-1} A^{s}(A-I)\left(\Xh\left(t-s+1\right)-X(t-s)\right)\right]_{i}}{\left[\left(A^{t}-\bm{\phi} \one^{T}\right) \one\right]_{i}+\bm{\phi}_{i} n} .
\end{split}
\end{align}
Using this, we find the following expression for the vector representing error of node $i$ :
\begin{align*} 
&\z_{i}(t+1)-\frac{\one^{T}X(1)}{n}\\[5pt]
&=\frac{\Big[A^{t} X(1)+\sum_{s=0}^{t-1} A^{s}(A-\mathbf{I})(\Xh(t-s+1)-X(t-s))\Big]_{i}}{n\left(\Big[(A^{t}-\bm{\phi} \one^{T})\one\Big]_i + \bm{\phi}_{i} n\right)}  -\frac{\one^{T} X(1)\left(\left[\left(A^{t}-\bm{\phi} \one^{T}\right) \one\right]_{i}+\bm{\phi}_{i} n\right)}{n\left(\left[(A^{t}-\bm{\phi} \one^{T}) \one\right]_{i}+\bm{\phi}_{i} n\right)}\\[5pt]
&=\frac{n\left[A^{t}-\bm{\phi}\one^T\right]_iX(1)+n \sum_{s=0}^{t-1}\left[
A^{s}(A-\mathbf{I})\right]_{i}(\Xh(t-s+1)-X(t-s))-\one^T X(1)\left[(A^{t}-\bm{\phi}\one^T)\right]_{i} \one}{n\left(\left[A^{t}-\bm{\phi}\one^T\right]_{i} \one + \bm{\phi}_{i} n\right)}.
\end{align*}
By Proposition \ref{lem:propo} it holds that $A^t\one\ge\delta$ and $\Big\|[A^t-\bm{\phi}\one^T]_i\Big\| \le C\la^t$ for all $t\ge1$; thus, we derive the following for the error of parameter of node $i$ :  
\begin{align*}
\mathbb{E} \bigg\|  \z_{i}(t+1)-\frac{\one^{T}X(1)}{n} \bigg\| 
&\leq \frac{C}{\delta} \cdot \lambda^{t}\Big\|X(1)\Big\| \quad+\frac{C}{\delta} \sum_{s=0}^{t-1} \lambda^{s} \cdot\left(\xi_1 \cdot \lambda^{\frac{t-s}{2}}\right) 
+ \frac{C}{n \delta} \lambda^{t}\Big\|\one^T X(1)\Big\| \\
&\leq \frac{C}{\delta} \cdot \lambda^{t}\Big\|X(1)\Big\|+\frac{C}{n \delta} \lambda^{t}\Big\|\one^{T} X(1)\Big\| \cdot \sqrt{n} 
+\frac{C}{\delta} \xi_1 \cdot \frac{\lambda^{t / 2}}{1-\lambda^{1 / 2}}.
\end{align*}
Note that $\| \one^{T} X(1)\|\leq\|X(1)\| \cdot \sqrt{n}$. Thus, 
\begin{equation*}
\begin{array}{l}
{\mathbb{E}\bigg\|\z_{i}(t+1)-\frac{\one^TX(1)}{n}\bigg\| \leq \frac{2C}{\delta} \lambda^{t}\Big\|X(1)\Big\|}
{+\frac{C}{\delta} \cdot \frac{\xi_1}{1-\lambda^{1/2}} \lambda^{t/ 2}}.
\end{array}
\end{equation*}
Rewriting the condition on $\la_1$ in Lemma \ref{lem:Uineq} based on $\omega$, we derive the inequality in the statement of Theorem \ref{thm:gossip}. 
\section{Proof of Theorem \ref{thm:convex} : Quantized Push-sum with Convex Objectives} 
First we write the iterations of Algorithm \ref{alg:Opt} in matrix notation to derive the following: 

\begin{align}\label{alg:Opt_matrix}
\left\{
      \begin{array}{ll}
 &Q(t)=Q\left(X(t)-\Xh(t)\right)\\
&\Xh(t+1)=\Xh(t)+Q(t)\\
&W(t+1)=X(t)+(A-I) \Xh(t+1)\\
&\y(t+1)=A \y(t)\\
 &\z_{i}(t+1)=\frac{\w _i(t+1)}{y_i(t+1)}\\
 &X(t+1)=W(t+1)-\alpha \,\partial F(Z(t+1),\zeta_{t+1})\\
      \end{array}
      \right.
\end{align}
Similar to \eqref{eq:updateXt}, we can rewrite the iterations of $X(t)$ to obtain the following expression: 
\begin{dmath}\label{eq:nonconvexxone}
\Big\|X(t+1)-\bm\phi \one^T X(t)\Big\|^{2} \leq 3 C^{2} \lambda^{2 t}\Big\|X(1)\Big\|^{2}+3 C^{2} \alpha^{2}\left\|\sum_{s=0}^{t} \lambda^{s}\Big\| \partial F(Z(t-s+1),\zeta_{t-s+1})\Big\|\right\|^{2} +3 C^{2}\left\|\sum_{s=0}^{t-1} \lambda^{s}\Big\| \Xh(t-s+1)-X(t-s)\Big\|\right\|^{2}
\end{dmath}

For the last term, we expand the summation as well as simplifying the resulting expressions  to obtain:   
\begin{align*}
&\left\|\sum_{s=0}^{t-1} \lambda^{s}\Big\| \Xh(t-s+1)-X(t-s)\Big\|\right\|^{2} \\
&=\sum_{s=0}^{t-1} \lambda^{2s}\Big\|\Xh(t-s+1)-X(t-s)\Big\|^{2}+\sum_{s \neq s^{\prime}} \lambda^{s} \cdot \lambda^{s^{\prime}}\Big\|\Xh(t-s+1)-X(t-s)\Big\|\Big\|\Xh\left(t-s^{\prime}+1\right)-X\left(t-s^{\prime}\right)\Big\|\\
\end{align*}
Using the relation $x\cdot y \le x^2/2 +y^2/2$ for all $x,y \in \R$, as well as the inequality $\la^{2s} \le \frac{\la^s}{1-\la}$, from the equation above the following yields: 
\begin{align*}
&\bigg\|\sum_{s=0}^{t-1} \lambda^{s}\Big\| \Xh(t-s+1)-X(t-s)\Big\|\bigg\|^{2} \\
&\le \sum_{s=0}^{t-1} \lambda^{2s}\Big\|\Xh(t-s+1)-X(t-s)\Big\|^{2} 
+ 1/2 \sum_{s \neq s^{\prime}} \lambda^{s+s'} \Big\|\Xh(t-s+1)-X(t-s)\Big\|^2 
\\
&\qquad +1/2 \sum_{s \neq s^{\prime}} \lambda^{s} \cdot \lambda^{s^{\prime}}\Big\|\Xh(t-s'+1)-X(t-s')\Big\|^2 \\
&\le  \sum_{s=0}^{t-1}\left(\lambda^{2 s}+\frac{\lambda^{s}}{1-\lambda}\right)\Big\|\Xh(t-s+1)-X(t-s)\Big\|^{2} \\
&\le\sum_{s=0}^{t-1} \frac{2 \lambda^{s}}{1-\lambda}\Big\|\Xh(t-s+1)-X(t-s)\Big\|^{2}.
\end{align*}
Using a similar approach as above, we can bound the second term in the RHS of \eqref{eq:nonconvexxone} as follows:
\begin{align*}
\left\|\sum_{s=0}^{t} \lambda^{s}\Big\| \partial F(Z(t-s+1),\zeta_{t-s+1})\Big\|\right\|^{2}  \le\sum_{s=0}^{t} \frac{2 \lambda^{s}}{1-\lambda}\Big\|\partial F(Z(t-s+1),\zeta_{t-s+1}) \Big\|^{2}.
\end{align*}
Thus from \eqref{eq:nonconvexxone} with the assumption that $X(1) = 0$ (as in Assumption \ref{assumption:initialization}) we find that
\begin{align}\label{eq:nonconvexxtwo}
\E&\Big\|X(t+1)-\bm\phi \one^T X(t)\Big\|^{2}\nn \\
&\le 6C^2\E\sum_{s=0}^{t-1} \frac{\lambda^{s}}{1-\lambda}\Big\|\Xh(t-s+1)-X(t-s)\Big\|^{2} + 6C^2\alpha^2\sum_{s=0}^{t-1} \frac{\lambda^{s}}{1-\lambda}\E\,\Big\|\partial F(Z(t-s+1),\zeta_{t-s+1}) \Big\|^{2} \nn\\&\le 6C^2\sum_{s=0}^{t-1} \frac{\lambda^{s}}{1-\lambda}\E\,\Big\|\Xh(t-s+1)-X(t-s)\Big\|^{2} + \frac{6C^2\alpha^2nD^2}{(1-\la)^2},
\end{align}
where we used Assumption \ref{assumption:boundedgrad} to bound the second term in the RHS of \eqref{eq:nonconvexxtwo}. Moreover, similar to the chain of inequalities in \eqref{eq:xt-xht} derived for Algorithm \ref{alg:gossip_matrix}, we derive the following inequality here in the presence of stochastic gradients:
\begin{align} \label{eq:nonconvexxhone}
&\mathbb{E}\Big\|X(t+1)-\Xh(t+2)\Big\|^{2} \nonumber\\
&\leq 3\,\omega^2\left(1+\gamma^{2}\right) \mathbb{E}\Big\|X(t)-\Xh(t+1)\Big\|^{2}
+3\,\omega^2 \gamma^{2} \mathbb{E}\Big\|X(t)-\phi \one^T X(t-1)\Big\|^{2}\nonumber\\
&\qquad +3\,\omega^2\alpha^2\E\Big\|\partial F(Z(t+1),\zeta_{t+1})\Big\|^{2}.
\end{align}
Let $\la_2:=\omega^2(1+\gamma^2)$. Then, from \eqref{eq:nonconvexxhone} as well as bounding stochastic gradients according to Assumption \ref{assumption:boundedgrad}, it directly follows that 
\begin{align}\label{eq:nonconvexxhtwo}
\E\Big\|X(t+1)-\Xh(t+2)\Big\|^{2} \leq 3\la_2\left( \E\Big\|X(t)-\Xh(t+1)\Big\|^{2} + \E\Big\|X(t)-\bm\phi \one^T X(t-1)\Big\|^{2}+ n\alpha^2 D^2\right).
\end{align}
Let $R(t+1) := \E \Big\|X(t+1)-\bm\phi \one^T X(t)\Big\|^{2}$ and $U(t+1):= \E\Big\|X(t+1)-\Xh(t+2)\Big\|^2 $. Then noting \eqref{eq:nonconvexxtwo} and \eqref{eq:nonconvexxhtwo}, we can rewrite the corresponding errors as follows: 
\begin{align}\label{eq:UandRnonconvex}
\left\{\begin{array}{l}{ \E R(t+1) \leq \frac{6C^{2} n \alpha^{2} D^2}{(1-\lambda)^{2}}+\frac{6 C^{2}}{1-\lambda} \sum_{s=0}^{t-1} \lambda^{s} U(t-s)}, \\[5pt] {\mathbb{E}U(t+1) \leq 3 \lambda_{2}\left(U(t)+R(t)+n \alpha^{2} D^2\right).}\end{array}\right.
\end{align}
Next, we state the following lemma which shows that the error of quantization i.e. $U(t)$ decays proportionately with $\alpha^2$ . 
\begin{lem}\label{lem:iteratesnonconvex}
Under Assumption \ref{assumption:initialization}, the inequalities in  \eqref{eq:UandRnonconvex} satisfy the following for all $\la_2 \le  (\frac{1}{6}+ \frac{C^2}{(1-\la)^2})^{-1}$ and all iterations $t\ge 1$ 
\begin{align}
U(t)\le\xi_2\,\al^2,
\end{align}
where $\xi_2 = 6 \la_2\left(nD^2+ \frac{6nC^2D^2}{(1-\la)^2}\right)$.
\end{lem}
\begin{proof}
First we write the inequalities in \eqref{eq:UandRnonconvex} based on $U(\cdot)$ to obtain: 
\begin{equation*}
U(t+1) \leq 3 \lambda_{2}\Big(U(t)+n \alpha^{2} D^2+\frac{6 nC^{2} \alpha^{2} D^2}{(1-\lambda)^{2}}+\frac{6 C^{2}}{1-\lambda} \sum_{s=0}^{t-2} \lambda^{s}U(t-s-1)\Big).
\end{equation*}
Thus, by the assumption of induction we obtain the following for $U(t+1)$: 
\begin{equation}
\begin{aligned}
U(t+1)&\le3 \lambda_{2}\left(\xi_2 \alpha^2+n \alpha^{2} D^2+\frac{6 C^{2} \alpha^{2} n D^2}{(1-\lambda)^{2}}+\frac{6 C^{2} \xi_2 \alpha^2}{(1-\lambda)} \sum_{s=0}^{t-2} \lambda^{s}\right) \\
&\le 3\la_2\al^2\xi_2\left(1+\frac{6C^2}{(1-\la)^2}\right) + 3\la_2\al^2\left(nD^2 + \frac{6nC^2D^2}{(1-\la)^2}\right)\\
&\le \frac{\xi_2 \,\al^2}{2} + 3\la_2\al^2\left(nD^2 + \frac{6nC^2D^2}{(1-\la)^2}\right) = \xi_2\,\al^2,
\end{aligned}
\end{equation}
where the last two steps follow from the assumptions for $\la_2$ and $\xi_2$, respectively. Note that based on iterations of the algorithm and Assumption \ref{assumption:initialization} we conclude that $\|U(1)\| = 0$. This completes the proof of the lemma. 
\end{proof}
\begin{lem}\label{lem:consensus_err}
Under Assumptions \ref{assumption:graph}-\ref{assumption:boundedvar}, if $\la_2 \le  (\frac{1}{6}+ \frac{C^2}{(1-\la)^2})^{-1}$, the following relation holds for the consensus error of Algorithm \ref{alg:Opt} for all $t\ge1$: 
\begin{align*}
\mathbb{E}\left\|\z_{i}(t+1)-\frac{\one^T X(t)}{n}\right\|^{2} \le \frac{6\, C^2\alpha^2}{\delta^2(1-\la)^2} (2nD^2 + \xi_2),
\end{align*}
where $\xi_2 = 6 \la_2\left(nD^2+ \frac{6nC^2D^2}{(1-\la)^2}\right)$.
\end{lem}
\begin{proof}
Using the update rule in Algorithm \ref{alg:Opt_matrix} and similar to \eqref{eq:z} we derive the following for the vector corresponding to consensus error of node $i$ :
\begin{align*}
&\z_{i}\left(t+1\right)-\frac{\one^T X(t)}{n}= \\
&\frac{\Big[\sum_{s=0}^{t-1} A^{s}(A-I)(\Xh(t-s+1)-X(t-s))-\alpha \sum_{s=0}^{t} A^s \partial F\Big(Z(t-s+1),\zeta_{t-s+1}\Big)\Big]_{i} }{\left[(A^{t}-\bm\phi \one^T) \one\right]_{i}+\bm\phi_{i}n} \\ 
&+\frac{\alpha \one^T \sum_{s=0}^{t-1} \partial F\Big(Z(t-s+1),\zeta_{t-s+1}\Big)\left(\left[\left(A^{t}-\bm\phi \one^T\right) \one\right]_{i}+\bm\phi_{i}n\right)}{n\left(\left[(A^{t}-\bm\phi \one^T) \one\right]_{i}+\bm\phi_{i}n\right)}.
\end{align*}
Note that by Proposition \ref{lem:propo} for all $t\ge1$ we have $[(A^{t}-\bm\phi \one^T) \one]_{i}+\bm\phi_{i}n = [A^t \one]_i \ge \delta$,  which yields the following for squared norm of consensus error:
\begin{align}\label{eq:consensuserror}
\begin{split}
 \mathbb{E}&\left\|\z_{i}(t+1)-\frac{\one^T X(t)}{n}\right\|^{2} 
\leq \frac{3}{\delta^{2}} \mathbb{E}\bigg\|\sum_{s=0}^{t}\Big[A^{s}(A-I)\Big]_{i}\Big(\Xh(t-s+1)-X(t-s)\Big)\bigg\|^{2} \\
+&\frac{3 \alpha^{2}}{\delta^{2}} \mathbb{E}\left\|\sum_{s=0}^{t}\left[A^{s}-\bm\phi \one^T\right]_{i} \partial F\Big(Z(t-s+1),\zeta_{t-s+1}\Big)\right\|^{2}\\
+&\frac{3 \alpha^{2}}{n^{2} \delta^{2}} \mathbb{E}\left\|\one^T\left(\sum_{s=0}^{t-1}\left[A^{t}-\bm\phi \one^T\right]_{i} \partial F\Big(Z\left(t-s+1\right),\zeta_{t-s+1}\Big)\right)\right\|^{2} .
\end{split}
\end{align}
By expanding the first term in the RHS of \eqref{eq:consensuserror} we derive
\begin{align*}
\begin{split}
&\left\|\sum_{s=0}^{t}\Big[A^{s}(A-I)\Big]_{i}\Big(\Xh(t-s+1)-X(t-s)\Big)\right\|^{2} =\sum_{s=0}^{t}\bigg\|\Big[A^{s}(A-I)\Big]_{i}\Big(\Xh(t-s+1)-X(t-s)\Big)\bigg\|^{2}\\
&+\sum_{s \neq s^{\prime}}^{t}\bigg\langle\Big[A^{s}(A-I)\Big]_{i}\Big(\Xh(t-s+1)-X(t-s)\Big),\Big[A^{s}(A-I)\Big]_{i}\left(\Xh\left(t-s^{\prime}+1\right)-X\left(t-s^{\prime}\right)\right)\bigg\rangle \\
&\le \sum_{s=0}^{t}\bigg\|\Big[A^{s}(A-I)\Big]_{i}\bigg\|^2\bigg\|\Xh(t-s+1)-X(t-s)\bigg\|^{2}\\
&+\sum_{s \neq s^{\prime}}^{t}\bigg\|\Big[A^{s}(A-I)\Big]_{i}\bigg\|\bigg\|\Xh(t-s+1)-X(t-s)\bigg\|\bigg\|\Big[A^{s}(A-I)\Big]_{i}\bigg\|\bigg\|\left(\Xh\left(t-s^{\prime}+1\right)-X\left(t-s^{\prime}\right)\right)\bigg\|.
\end{split}
\end{align*}
Using the relation $x\cdot y \le x^2/2 +y^2/2$ for all $x,y \in \R$, this inequality reduces to the following: 
\begin{align*}
\begin{split}
&\left\|\sum_{s=0}^{t}\Big[A^{s}(A-I)\Big]_{i}\Big(\Xh(t-s+1)-X(t-s)\Big)\right\|^{2} \nonumber\\
&\le  \sum_{s=0}^{t}\bigg\|\Big[A^{s}(A-I)\Big]_{i}\bigg\|^2\bigg\|\Xh(t-s+1)-X(t-s)\bigg\|^{2}\\
&+ \frac{1}{2}\sum_{s \neq s^{\prime}}^{t}\bigg\|\Big[A^{s}(A-I)\Big]_{i}\bigg\|\bigg\|\Big[A^{s'}(A-I)\Big]_{i}\bigg\|\left(\bigg\|\Xh(t-s+1)-X(t-s)\bigg\|^2 + \bigg\|\Xh(t-s'+1)-X(t-s')\bigg\|^2\right).
\end{split}
\end{align*}
Next we use Proposition \ref{lem:propo} to yield that 
\begin{align*}
\begin{split}
&\left\|\sum_{s=0}^{t}\Big[A^{s}(A-I)\Big]_{i}\Big(\Xh(t-s+1)-X(t-s)\Big)\right\|^{2}\le \\
&C^2 \sum_{s=0}^{t}\la^{2s}\bigg\|\Xh(t-s+1)-X(t-s)\bigg\|^{2} + C^2\sum_{s\neq s'}^{t}\la^{s+s'}\bigg\|\Xh(t-s+1)-X(t-s)\bigg\|^{2}\\
&\le C^2 \sum_{s=0}^{t}\Big(\la^{2s}+\frac{\la^s}{1-\la}\Big)\bigg\|\Xh(t-s+1)-X(t-s)\bigg\|^{2} \\
&\le \frac{2C^2}{1-\la}\sum_{s=0}^t\la^s\bigg\|\Xh(t-s+1)-X(t-s)\bigg\|^{2} ,
\end{split}
\end{align*}
where we used $\la^{2s} \le \frac{\la^s}{1-\la}$ to derive the last inequality.
Using the same approach for the second term in the RHS of \eqref{eq:consensuserror}, we derive the following upper bound: 
\begin{align*}
    \mathbb{E}\left\|\sum_{s=0}^{t}\left[A^{s}-\bm\phi \one^T\right]_{i} \partial F\Big(Z(t-s+1),\zeta_{t-s+1}\Big)\right\|^{2}\le \frac{2 C^{2}}{1-\lambda} \sum_{s=0}^{t} \lambda^{s}\E\Big\|\partial F(Z(t-s+1),\zeta_{t-s+1})\Big\|^{2}.
    \end{align*}
To bound the third term in the RHS of \eqref{eq:consensuserror}, we use the same method,as well as the fact that $\|A^t-\bm \phi \one^T\| \le \la^t \le \la^s$ for all $s\le t$ to deduce that  
    \begin{align*}
    \mathbb{E}\left\|\one^T\left(\sum_{s=0}^{t-1}\left[A^{t}-\bm\phi \one^T\right]_{i} \partial F\Big(Z\left(t-s+1\right),\zeta_{t-s+1}\Big)\right)\right\|^{2}\le\frac{2n C^{2}}{1-\lambda} \sum_{s=0}^{t} \lambda^{s}\E\Big\|\partial F(Z(t-s+1),\zeta_{t-s+1})\Big\|^{2}.
\end{align*}

Replacing these back in \eqref{eq:consensuserror} gives   
\begin{align}\label{eq:consensuserror2}
 &\mathbb{E}\left\|\z_{i}(t+1)-\frac{\one^T X(t)}{n}\right\|^{2} \nonumber\\
 &\le\frac{6C^2}{\delta^2(1-\la)}\sum_{s=0}^t\la^s\E\bigg\|\Xh(t-s+1)-X(t-s)\bigg\|^{2} \nonumber\\
 &\qquad + \left(\frac{6\al^2C^2}{\delta^2(1-\la)} + \frac{6\al^2C^2}{n\delta^2(1-\la)}\right)\sum_{s=0}^{t} \lambda^{s}\E\Big\|\partial F(Z(t-s+1),\zeta_{t-s+1})\Big\|^{2}.
\end{align}
Note that $U(t-s):=\E\,\bigg\|\Xh(t-s+1)-X(t-s)\bigg\|^{2}\le\xi_2\al^2 $ by Lemma \ref{lem:iteratesnonconvex}. Also by bounded stochastic gradient property in Assumption \ref{assumption:boundedgrad}, we have $\E \,\Big\|\partial F(Z(t-s+1),\zeta_{t-s+1})\Big\|^{2}\le nD^2.$ Therefore we conclude the following from \eqref{eq:consensuserror2}:
\begin{align}\label{eq:cons_error}
\mathbb{E}\left\|\z_{i}(t+1)-\frac{\one^T X(t)}{n}\right\|^{2} \le \frac{6\,\xi_2\, C^2\alpha^2}{\delta^2(1-\la)^2} + \left(\frac{6\al^2C^2}{\delta^2(1-\la)^2} + \frac{6\al^2C^2}{n\delta^2(1-\la)^2}\right)nD^2.
\end{align}
Simplifying relations with $1+1/n\le 2$,  yields the desired inequality in the statement of the lemma.
\end{proof}
We continue with proving the next lemma which relates the consensus error as stated in Lemma \ref{lem:consensus_err} to the error of global objective function $f$ evaluated at the average of parameters $\x_i(t)$ of all nodes. 
\begin{lem}\label{lem:cons_to_f}
For all $t \geq 1$, iterations of Algorithm \ref{alg:Opt_matrix} satisfy:
\begin{align*}
\E\,\Big\|\bar X(t+1)-\z^\star\Big\|^{2} \leq\E\,\Big\|\bar{X}(t)-\z^\star\Big\|^{2}- &\left(2\alpha - \frac{8\alpha^2 L}{n}\right)\E(f\left(\bar{X}(t))-f(\z^\star)\right) \\ 
&+\frac{2\alpha L + 4L^2\alpha^2}{n} \sum_{i=1}^{n} \E\,\Big\|\z_{i}(t+1)-\bar{X}(t)\Big\|+\frac{2\sig^2\alpha^2}{n}.
\end{align*}
\end{lem}
\begin{proof}
First we recall that $A$ is column stochastic; thus, it yields that $\one^T\left(A-I\right) = 0$. Therefore, iterations of  Algorithm \ref{alg:Opt_matrix} yield that
\begin{align*}
\bar{X}(t+1)=\bar{X}(t)-\frac{\alpha}{n} \sum_{i=1}^{n} \nabla F_{i}(\z_i(t+1),\zeta_{i,t+1}).
\end{align*}
Thus, 
\begin{equation}\label{eq:lemmain}
\begin{aligned}
\E\,\left\|\bar{X}(t+1)-\z^\star\right\|^{2}=&\E\,\left\|\bar{X}(t)-\z^\star\right\|^{2} 
-\frac{2 \alpha}{n} \sum_{i=1}^{n} \E\,\Big\langle \nabla f_{i}\left(\z_i (t+1)\right),\bar{X}(t)-\z^\star\Big\rangle \\&+\alpha^{2}\,\E\,\Big\|\frac{1}{n}\sum_{i=1}^{n} \nabla F_{i}(\z_i (t+1),\zeta_{i,t+1})\Big\|^{2}.
\end{aligned}
\end{equation}
We derive an upper bound for the third term by adding and subtracting $\frac{1}{n}\sum_{i=1}^n\nabla f_i(\z_i(t+1))$ : 
\begin{align}
\E\Big\|\frac{1}{n}\sum_{i=1}^{n} \nabla F_{i}(\z_i (t+1),\zeta_{i,t+1})\Big\|^{2} &\le2\E \Big\|\frac{1}{n}\sum_{i=1}^{n} \left(\nabla F_{i}(\z_i (t+1),\zeta_{i,t+1})-\nabla f_i(\z_i(t+1))\right) \Big\|^{2} \nn\\&\quad + 2\E \Big\|\frac{1}{n}\sum_{i=1}^{n} \nabla f_i(\z_i(t+1))\Big\|^{2}\nn\\
&\le \frac{2}{n^2}\sum_{i=1}^n\E\Big\| \nabla F_{i}(\z_i (t+1),\zeta_{i,t+1})-\nabla f_i(\z_i(t+1))\Big\|^2\nn\\
&\quad +2\E \Big\|\frac{1}{n}\sum_{i=1}^{n} \nabla f_i(\z_i(t+1))\Big\|^{2}\le\frac{2\sig^2}{n} + 2\E \Big\|\frac{1}{n}\sum_{i=1}^{n} \nabla f_i(\z_i(t+1))\Big\|^{2}.
\label{eq:nablaf}
\end{align}
For the last term in \eqref{eq:nablaf} we have
\begin{align*}
\E \Big\|\frac{1}{n}\sum_{i=1}^{n} \nabla f_i(\z_i(t+1))\Big\|^{2} &\le 2\E \left\|\frac{1}{n}\sum_{i=1}^{n} (\nabla f_i(\z_i(t+1))-\nabla f_i(\bar{X}(t)))\right\|^{2} + 2\E \left\|\frac{1}{n}\sum_{i=1}^{n} (\nabla f_i(\bar{X}(t))-\nabla f_i(\z^\star))\right\|^{2}  \\ 
&\le \frac{2}{n}\sum_{i=1}^n\E\,\bigg\|\nabla f_i(\z_i(t+1))-\nabla f_i(\bar{X}(t))\bigg\|^2 + 2\E\bigg\|\nabla f(\bar{X}(t))-\nabla f(\z^\star)\bigg\|^2 \\
&\le \frac{2L^2}{n}\sum_{i=1}^n \E\Big\|\z_i(t+1)- \bar{X}(t)\Big\|^2 + \frac{4L}{n}(\E f(\bar{X}(t))-f(\z^\star)).
\end{align*}
Replacing this back in \eqref{eq:nablaf} we have 
\begin{align}\label{eq:nablaffinal}
\E\left\|\frac{1}{n}\sum_{i=1}^{n} \nabla F_{i}(\z_i (t+1),\zeta_{i,t+1})\right\|^{2} \le \frac{2\sig^2}{n} +  \frac{4L^2}{n}\sum_{i=1}^n \E\Big\|\z_i(t+1)- \bar{X}(t)\Big\|^2 + \frac{8L}{n}(\E f(\bar{X}(t))-f(\z^\star)).
\end{align}
Next, we will achieve a bound for the second term in the RHS of \eqref{eq:lemmain}. By adding and subtracting $\z_i(t+1)$  we have 
\begin{align}\label{eq:leml}
\E\Big\langle \nabla f_i(\z_i(t+1)),\bar{X}(t)-\z^\star\Big\rangle  =  \E\Big\langle \nabla f_i(\z_i(t+1)),\bar{X}(t)-\z_i(t+1)\Big\rangle + \E\Big\langle \nabla f_i(\z_i(t+1)),\z_i(t+1)-\z^\star\Big\rangle,
\end{align}
where after using $L$-smoothness of the function $f_i$ for the first term and convexity of the function $f_i$ for the second term of \eqref{eq:leml} we deduce that
\begin{align*}
\E\,\Big\langle \nabla f_i(\z_i(t+1)),\bar{X}(t)-\z^\star\Big\rangle  &\ge \E f_i(\bar{X}(t)) - \E f_i(\z_i(t+1)) - \frac{L}{2}\E\Big\|\bar{X}(t) - \z_i(t+1)\Big\|^2 + \E f_i(\z_i(t+1)) - f_i(\z^\star) \\
&=  \E f_i(\bar{X}(t)) - f_i(\z^\star)  - \frac{L}{2}\E\Big\|\bar{X}(t) - \z_i(t+1)\Big\|^2.
\end{align*}
Using this and recalling that $f(\cdot) = \frac{1}{n} \sum_{i=1}^n f_i(\cdot)$, we derive the following:
$$
-\frac{2\alpha}{n}\sum_{i=1}^n \E\,\Big\langle \nabla f_i(\z_i(t+1)),\bar{X}(t)-\z^\star\Big\rangle \leq -2\al(\E\,f(\bar{X}(t)) - f(\z^\star))+\frac{2\al L}{n} \sum_{i=1}^n \E\Big\|\bar{X}(t) - \z_i(t+1)\Big\|^2.
$$
By replacing this result and \eqref{eq:nablaffinal} in \eqref{eq:lemmain} we derive the desired inequality in the statement of the lemma.
\end{proof}
We continue with rearranging and summing both sides of Lemma \ref{lem:cons_to_f} for $t= 1 ,\cdots,T$ to find the following:
\begin{align*}
& \left(2 \alpha-\frac{8\alpha^2L}{n}\right)\sum_{t=1}^{T}(\E f(\bar{X}(t))-f(\z^\star))\nonumber\\
& \leq\Big \|\bar{X}(1)-\z^\star\Big\|^2 +  \frac{2\alpha L + 4\alpha^2 L^2}{n} \sum_{t=1}^T\sum_{i=1}^n \E\,\Big\|\z_i(t+1)-\bar{X}(t)\Big\|^2 +\sum_{t=1}^T \frac{2\sig^2\alpha^2}{n}.
\end{align*}
Scaling both sides, as well as noting the Assumption \ref{assumption:initialization} i.e. $X(1) = 0$,  we find that  
\begin{align*}
&\frac{1}{T}\sum_{t=1}^{T}(\E f(\bar{X}(t))-f(\z^\star)) \\
&\leq\frac{1}{2T\al(1-\frac{4\al L}{n})}\left \|  \z^\star \right\|^2 + \frac{2\al L^2 + L}{nT(1-\frac{4\al L}{n})}\sum_{t=1}^T \sum_{i=1}^n \E\Big\|\z_i(t+1)-\bar{X}(t)\Big\|^2 + \frac{\alpha \sig^2}{n(1-\frac{4\al L}{n})}.
\end{align*}
Using convexity of $f(\cdot)$, the above inequality simplifies to the following:
\begin{align*}
&\E\,f\left(\frac{1}{T}\sum_{t=1}^T\bar{X}(t)\right)-f(\z^\star)
 \\& \leq \frac{1}{2T\al(1-\frac{4\al L}{n})}\left \|  \z^\star \right\|^2 + \frac{2\al L^2 + L}{nT(1-\frac{4\al L}{n})}\sum_{t=1}^T \sum_{i=1}^n \E\Big\|\z_i(t+1)-\bar{X}(t)\Big\|^2 + \frac{\alpha \sig^2}{n(1-\frac{4\al L}{n})},
\end{align*}
which after replacing the consensus error from Lemma \ref{lem:consensus_err}, further simplifies into 
\begin{align}\label{eq:ave}
&\E f\left(\frac{1}{T}\sum_{t=1}^T\bar{X}(t)\right)-f(\z^\star) \nonumber\\
&\leq\frac{1}{2T\al(1-\frac{4\al L}{n})}\left \|  \z^\star \right\|^2   + 
 \frac{6\al^2C^2\left(2nD + \xi_2\right)}{\delta^2(1-\la)^2} \cdot \frac{2\al L^2 + L}{1-\frac{4\al L}{n}} + \frac{\alpha \sig^2}{n(1-\frac{4\al L}{n})}.
\end{align}
The inequality in \eqref{eq:ave} guarantees the convergence of time average of $\bar{X}(t) = \frac{1}{n}\sum_{i=1}^n \x_i(t) $ to the optimal point $\z^\star$. However computing the average of $\x_i(t)$ between workers in every iteration is time consuming since it can not be done in the decentralized setting. Next we show that the time average of the local variables, $\z_i(t)$ converges to an optimal $\z^\star$,  for every node $i$ . First, By $L$-smoothness of the function $f(\cdot)$ as well as the inequality $\langle \x,\y\rangle \le \frac{1}{2}\|\x\|^2 + \frac{1}{2}\|\y\|^2$, we derive for all $i \in [n]$ it holds that
\begin{align}
&f\left( \frac{1}{T} \sum_{t=1}^T \z_i(t+1)\right)- f\left( \frac{1}{T} \sum_{t=1}^T \bar{X}(t)\right)  \nn\\
&\le \left\langle \frac{1}{T}\sum_{t=1}^T \z_i(t+1) - \bar{X}(t), \nabla f\left(\frac{1}{T} \sum_{t=1}^T \bar{X}(t)\right)  \right\rangle + \frac{L}{2T^2} \left\|\sum_{t=1}^T\z_i(t+1)-\bar{X}(t)\right\|^2  \nn\\
&\le \frac{1}{2T^2}\left\|\sum_{t=1}^T\z_i(t+1)-\bar{X}(t) \right\|^2 + \frac{1}{2}\left\|\nabla f\left(\frac{1}{T} \sum_{t=1}^T \bar{X}(t)\right) \right\|^2 +  \frac{L}{2T^2} \left\|\sum_{t=1}^T\z_i(t+1)-\bar{X}(t)\right\|^2  .\label{eq:subtract}
\end{align}
Note that for the optimal solution $\z^\star$ it holds that $\|\nabla f(\x)\|^2=\|\nabla f(\x) - \nabla f(\z^\star)\|^2 \le 2L \left(f(\x) - f(\z^\star)\right)$. Using this inequality for the second term in \eqref{eq:subtract} we conclude that  
\begin{align}\label{eq:lip}
&\E\,f\left( \frac{1}{T} \sum_{t=1}^T \z_i(t+1)\right)- \E\,f\left( \frac{1}{T} \sum_{t=1}^T \bar{X}(t)\right)  \nn \\
&\le \left(\frac{1}{2T} + \frac{L}{2T}\right)\sum_{t=1}^T\E\,\Big\|\z_i(t+1)-\bar{X}(t) \Big\|^2 + L\left(\E\,f\left( \frac{1}{T} \sum_{t=1}^T \bar{X}(t)\right) - f(\z^\star)\right).
\end{align}
By replacing the consensus error as derived in \eqref{eq:cons_error} and combining the inequalities \eqref{eq:ave} and \eqref{eq:lip}, we get the convergence error of the time average of local variables $\z_i$ :
\begin{align}\label{eq:convexfinal}
&\E f\left(\frac{1}{T}\sum_{t=1}^T\z_i(t+1)\right) - f(\z^\star) \nonumber\\
&\le \;\frac{L+1}{2T\al(1-\frac{4\al L}{n})}\left \|  \z^\star \right\|^2   + \left(\frac{(2\al L^2 + L)(L+1)}{1-\frac{4\al L}{n}}+\frac{1+L}{2} \right) \left( \frac{6\al^2C^2\left(2nD + \xi_2\right)}{\delta^2(1-\la)^2}\right) + \frac{\alpha \sig^2 (L+1)}{n(1-\frac{4\al L}{n})} .
\end{align}
We choose $\alpha= \frac{\sqrt{n}}{8L\sqrt{T}}$ in \eqref{eq:convexfinal}, which results in $\left(1-\frac{4\al L}{n}\right)^{-1}\le2$ for all $T\ge1$ and $n\ge1$. Furthermore, 
\begin{align}\label{eq:convexfinal2}
&\E f\left(\frac{1}{T}\sum_{t=1}^T\z_i(t+1)\right) - f(\z^\star) \nonumber\\
&\le \;\frac{8L(L+1)}{\sqrt{nT}}\left \|  \z^\star \right\|^2   + \frac{\sig^2(L+1)}{4\,L\sqrt{nT}} +\frac{C^2 n \left(\left(L+1\right)\left(\frac{L\sqrt{n}}{2\sqrt{T}}+L+1\right)\right) \left(\xi_2+2nD^2\right)}{10\,T\delta^2(1-\la)^2L^2 }.
\end{align}
Replacing $\xi_2$ and rewriting the condition on $\la_2$ with $\omega$ (as required by Lemmas \ref{lem:iteratesnonconvex} and \ref{lem:consensus_err} ) completes the proof. 
\section{Proof of Theorem \ref{thm:nonconvex} : Quantized Push-sum with Non-convex Objectives}
Using the $L$-smoothness of the global objective function which is implied by Assumption \ref{assumption:lipshitz}, we have for all $t\ge1$:
\begin{align}\label{eq:lsmooth}
\begin{split}
\mathbb{E} f\left(\frac{\one^T X(t+1)}{n}\right)=&\,\mathbb{E}\, f\left(\frac{\one^T X(t)}{n}-\alpha \frac{\one^T \partial F(Z(t+1,\zeta_{t+1}))}{n}\right)\\
= &\,\mathbb{E}\, f\left(\frac{\one^T X(t)}{n}\right)-\alpha\,\mathbb{E}\left\langle\nabla f\left(\frac{\one^TX(t)}{n}\right), \frac{\one^T \partial f(Z(t+1))}{n}\right\rangle \\
&+\frac{\alpha^{2} L}{2}\mathbb\, \E\left\|\frac{\one^T \partial F(Z(t+1,\zeta_{t+1}))}{n}\right\|^{2}.
\end{split}
\end{align}
For the last term in the RHS of \eqref{eq:lsmooth} we add and subtract $\frac{\one^T}{n} \partial f(Z(t+1))$ to yield 
\begin{align*}
\mathbb{E} \left\| \frac{\one^T \partial F(Z(t+1),\zeta_{t+1})}{n}
\right\|^2=&
\mathbb{E} \left\|\frac{\one^T}{n}\Big(\partial F(Z(t+1),\zeta_{t+1})-\partial f(Z(t+1))\Big)+\frac{\one^T}{n} \partial f(Z(t+1))\right\|^{2} \\
=&\mathbb{E}\left\|\frac{\one^T}{n}\Big(\partial F(Z(t+1),\zeta_{t+1})-\partial f(Z(t+1))\Big)\right\|^2 + \mathbb{E}\left\|\frac{\one^T}{n} \partial f(Z(t+1))\right\|^{2} \\
&\quad +2\, \mathbb{E} \left\langle \frac{\one^T}{n}\Big(\partial F(Z(t+1),\zeta_{t+1})-\partial f(Z(t+1))\Big), \frac{\one^T}{n} \partial f\left(Z(t+1)\right)\right\rangle .
\end{align*}
Since stochastic gradients of all nodes are unbiased estimators of the local gradients, the last term is zero in expectation. Thus,
\begin{align*}
\mathbb{E} \left\| \frac{\one^T \partial F(Z(t+1),\zeta_{t+1})}{n}
\right\|^2 \le\,&\mathbb{E}\left\|\frac{\one^T}{n}\Big(\partial F(Z(t+1),\zeta_{t+1})-\partial f(Z(t+1))\Big)\right\|^2 + \mathbb{E}\left\|\frac{\one^T}{n} \partial f(Z(t+1))\right\|^{2} \\+&
2\mathbb{E} \left\langle \E_{\zeta_{t+1}}\frac{\one^T}{n}\Big(\partial F(Z(t+1),\zeta_{t+1})-\partial f(Z(t+1))\Big), \frac{\one^T}{n} \partial f\left(Z(t+1)\right)\right\rangle  \\
=&\,\mathbb{E}\left\|\frac{\one^T}{n}\Big(\partial F(Z(t+1),\zeta_{t+1})-\partial f(Z(t+1))\Big)\right\|^2 + \mathbb{E}\left\|\frac{\one^T}{n} \partial f(Z(t+1))\right\|^{2} .
\end{align*}
Next, by expanding and using the fact that stochastic gradients are computed independently among different nodes, we show that the first term in the equation above is bounded: 
\begin{align*}
&\mathbb{E}\left\|\frac{\one^T}{n}\Big(\partial F(Z(t+1),\zeta_{t+1})-\partial f(Z(t+1))\Big)\right\|^2 =\frac{1}{n^{2}} \mathbb{E}\left\|\sum_{i=1}^{n} \nabla F_{i}\left(\z_{i}(t+1),\zeta_{i,t+1}\right)-\nabla f_{i}\left(\z_{i}(t+1)\right)\right\|^{2}\\
&=\frac{1}{n^{2}} \mathbb{E} \sum_{i=1}^{n}\Big\|\nabla F_{i}\left(\z_{i}(t+1),\zeta_{i,t+1}\right)-\nabla f_{i}\left(\z_{i}(t+1)\right)\Big\|^{2} \\
&+ \frac{1}{n^2} \E \sum_{i\neq i'}\bigg\langle \nabla F_{i}\left(\z_{i}(t+1),\zeta_{i,t+1}\right)-\nabla f_{i}\left(\z_{i}(t+1)\right), \nabla F_{i'}\left(\z_{i'}(t+1),\zeta_{i',t+1}\right)-\nabla f_{i'}\left(\z_{i'}(t+1)\right)\bigg\rangle\\
&= \frac{1}{n^{2}} \mathbb{E} \sum_{i=1}^{n}\Big\|\nabla F_{i}\left(\z_{i}(t+1),\zeta_{i,t+1}\right)-\nabla f_{i}\left(\z_{i}(t+1)\right)\Big\|^{2} \\
&+ \frac{1}{n^2} \E \sum_{i\neq i'}\bigg\langle\E_{\zeta_{i,t+1}} \nabla F_{i}\left(\z_{i}(t+1),\zeta_{i,t+1}\right)-\nabla f_{i}\left(\z_{i}(t+1)\right), \nabla F_{i'}\left(\z_{i'}(t+1),\zeta_{i',t+1}\right)-\nabla f_{i'}\left(\z_{i'}(t+1)\right)\bigg\rangle \\
&= \frac{1}{n^{2}} \mathbb{E} \sum_{i=1}^{n}\Big\|\nabla F_{i}\left(\z_{i}(t+1),\zeta_{i,t+1}\right)-\nabla f_{i}\left(\z_{i}(t+1)\right)\Big\|^{2}  \le \frac{\sig^2}{n},
\end{align*}
where we recall Assumption \ref{assumption:boundedvar} in the last inequality. Next, we rewrite \eqref{eq:lsmooth} using the new terms as follows: 
\begin{align}\label{eq:lsmooth2}
\begin{split}
\mathbb{E} f\left(\frac{\one^T X(t+1)}{n}\right) &\le \mathbb{E} f\left(\frac{\one^T X(t)}{n}\right) -\alpha\,\mathbb{E}\left\langle\nabla f\left(\frac{\one^TX(t)}{n}\right), \frac{\one^T \partial f(Z(t+1))}{n}\right\rangle \\ 
&\qquad \qquad +\frac{\al^2L}{2} \left(\frac{\sig^2}{n} + \mathbb{E}\left\|\frac{\one^T}{n} \partial f(Z(t+1))\right\|^{2} \right).
\end{split}
\end{align}

Moreover, using the relation $\langle \x,\y\rangle= \frac{1}{2}\|\x\|^2 +  \frac{1}{2}\|\y\|^2 - \frac{1}{2} \|\x-\y\|^2$, we find the following for the second term in the RHS of \eqref{eq:lsmooth2}:
\begin{align}\label{eq:cross}
\begin{split}
&\mathbb{E}\left\langle\nabla f\left(\frac{\one^TX(t)}{n}\right), \frac{\one^T \partial f(Z(t+1))}{n}\right\rangle \\ =& \,\frac{1}{2} \E\,\left\| \nabla f\left(\frac{\one^TX(t)}{n}\right)\right\|^2 + \frac{1}{2} \E \,\left\|  \frac{\one^T \partial f(Z(t+1))}{n}\right\|^2 - \frac{1}{2}\E\, \left\| \nabla f\left(\frac{\one^TX(t)}{n}\right)-   \frac{\one^T \partial f(Z(t+1))}{n}\right\|^2.
\end{split}
\end{align}
Using $L$-lipschitz assumption of local gradients (Assumption \ref{assumption:lipshitz}), the last term in \eqref{eq:cross} reduces to 
\begin{align*}
\mathbb{E}\left\|\nabla f\left(\frac{ \one^T X(t)}{n}\right)-\frac{\one^T \partial f(Z(t+1))}{n}\right\|^{2}
\leq &\,\frac{1}{n}\, \sum_{i=1}^{n}\E\,\left\|\nabla f_{i}\left(\frac{\one^T X(t)}{n}\right)-\nabla f_{i}\left(\z_{i}(t+1)\right)\right\|^{2} \\
\leq &\,\frac{L^{2}}{n} \sum_{i=1}^{n} \mathbb{E}\left\|\frac{\one^T X(t)}{n}-\z_{i}(t+1)\right\|^{2},
\end{align*}
where we used $\nabla f(\x) = \frac{1}{n}\sum_{i=1}^n\nabla f_i(\x)$ in the first step. Replacing this in \eqref{eq:cross} and finally substituting the resulting expression of \eqref{eq:cross} in \eqref{eq:lsmooth2} yields
\begin{align}\label{eq:lsmooth3}
\begin{split}
\mathbb{E} f\left(\frac{\one^T X(t+1)}{n}\right)& \leq \mathbb{E}f\left(\frac{\one^T X(t)}{n}\right)+\frac{\alpha^{2} \sigma^{2} L}{2 n}\\
&+\frac{\alpha^{2} L-\alpha}{2} \,\mathbb{E}\left\|\frac{\one^T \partial f(Z(t+1))}{n}\right\|^{2}-\frac{\alpha}{2} \,\mathbb{E}\left\|\nabla f\left(\frac{\one^T X(t)}{n}\right)\right\|^{2} \\
&+\frac{\alpha L^{2}}{n} \sum_{i=1}^{n} \mathbb{E}\left\|\z_{i}(t+1)-\frac{\one^T X(t)}{n}\right\|^{2}.
\end{split}
\end{align}

By replacing the value of consensus error from Lemma \ref{lem:consensus_err} in \eqref{eq:lsmooth3}, we obtain the following:
\begin{align}\label{eq:lsmooth4}
\begin{split}
\mathbb{E} f\left(\frac{\one^T X(t+1)}{n}\right)& \leq \mathbb{E}f\left(\frac{\one^T X(t)}{n}\right)+\frac{\alpha^{2} \sigma^{2} L}{2 n}\\
&+\frac{\alpha^{2} L-\alpha}{2} \,\mathbb{E}\left\|\frac{\one^T \partial f(Z(t+1))}{n}\right\|^{2}-\frac{\alpha}{2} \,\mathbb{E}\left\|\nabla f\left(\frac{\one^T X(t)}{n}\right)\right\|^{2} \\
&+\frac{6\al^3C^2L^2}{\delta^2(1-\la)^2}\left(nD^2+D^2 +\xi_2\right).
\end{split}
\end{align}
Thus by rearranging terms in \eqref{eq:lsmooth4} and averaging both sides from $t=1$ to $t=T$ we conclude that 
\begin{align}\label{eq:nonconvex_fixedal}
&\frac{1-\alpha L}{T}\sum_{t=1}^{T} \,\mathbb{E}\left\|\frac{\one^T \partial f(Z(t+1))}{n}\right\|^{2}+\frac{1}{T}\sum_{t=1}^T \,\mathbb{E}\left\|\nabla f\left(\frac{\one^T X(t)}{n}\right)\right\|^{2} \nonumber\\
& \le \frac{2\left(f\left(\frac{\one^TX(1)}{n}\right)-f^{\star}\right)}{\al T} 
+ \frac{\al\sig^2L}{n} + \frac{12\al^2L^2C^2}{\delta^2(1-\la)^2}\left(\xi_2+nD^2+D^2\right).
\end{align}
By choosing $\al = \frac{\sqrt{n}}{L\sqrt{T}}$ and noting that $X(1)=0$ we derive the desired statement of the theorem. Note that $T\ge4n$ implies that $\frac{1-\alpha L}{T}\ge\frac{1}{2T}$ guaranteeing that the first term in the LHS of \eqref{eq:nonconvex_fixedal} is positive. 
The consensus error in \eqref{eq:consen_thm} is concluded from Lemma \ref{lem:consensus_err} with the given choice of step-size as in the statement of the theorem. This completes the proof of the theorem.

\end{document}